\documentclass{LMCS}

\def\doi{7 (3:22) 2011}
\lmcsheading%
{\doi}
{1--31}
{}
{}
{Mar.~\phantom08, 2011}
{Sep.~29, 2011}
{}   
\usepackage{enumerate}
\usepackage{hyperref}

\usepackage{listings}
\usepackage{verbatim}

\usepackage{amsmath}
\usepackage{amsfonts}
\usepackage{amssymb}
\usepackage{calc}
\usepackage{multirow}
\usepackage{graphicx}
\usepackage{stmaryrd}
\usepackage{color}

\usepackage{pgf}
\usepackage{pgfarrows}
\usepackage{pgfnodes}
\usepackage{tikz}

\newenvironment{notation}
  {\begin{trivlist}\item[\hskip\labelsep{\bf Notation:}\small]}{\end{trivlist}}

\newcommand{{\msa}} {{\sf Msa}}
\newcommand{{\mwa}} {{\sf Mwa}}
\newcommand{{\msc}} {{\sf Msc}}
\newcommand{{\mnc}} {{\sf Mnc}}
\newcommand{{\hmsa}} {{\sf Hmsa}}
\newcommand{{\hmwa}} {{\sf Hmwa}}
\newcommand{{\hmsc}} {{\sf Hmsc}}
\newcommand{{\hmnc}} {{\sf Hmnc}}
\newcommand{{\bn}}{{\sf bin}}
\newcommand{{\wrd}}{{\sf wrd}}

\newcommand{\Det}{\rm Deterministic}
\newcommand{\alsure} {\mathbb{L}(\PBA^{=1})}
\newcommand{\alsureh} {\mathbb{L}(\HPBA^{=1})}
\newcommand{\probable} {\mathbb{L}(\PBA^{>0})}
\newcommand{\probableh} {\mathbb{L}(\HPBA^{>0})}
\newcommand{\strict} {\mathbb{L}(\PBA^{>\half})}
\newcommand{\stricth} {\mathbb{L}(\HPBA^{>\half})}
\newcommand{\nstrict} {\mathbb{L}(\PBA^{\geq\half})}
\newcommand{\nstricth} {\mathbb{L}(\HPBA^{\geq\half})}
\newcommand{\cBool}{\mathsf{BCl}}

\newcommand{\restrict}[2] {#1|_{#2}}

\newcommand{\Path}{\mathsf{Paths}}

\renewcommand{\subset}  {\subseteq}

\newcommand{\Good} {\mathsf{Good}}
\newcommand{\Bad} {\mathsf{Bad}}

   \newenvironment{claim}
  {\begin{trivlist}\item[\hskip\labelsep{\bf Claim:}\small]}{\end{trivlist}}

\newcommand{\fpm}[1]                         {(#1)}
\newcommand{\qs}                             {q_s}
\newcommand{\Q}                              {Q}
\newcommand{\Qf}                              {Q_f}
\newcommand{\qr}                             {q_r}
\newcommand{\qa}                             {q_a}
\newcommand{\unit}                           {[0,1]}
\newcommand{\cM}                             {\mathcal{M}} 
\newcommand{\FPM}                            {{\rm FPM}} 
\newcommand{\HPM}                            {{\rm HPM}}

\newcommand{\RFPM}                           {{\rm RatFPM}} 
\newcommand{\RHPM}                           {{\rm RatHPM}} 
\newcommand{\st}                             {\mathbin{|}}
\newcommand{\PFA}                            {{\rm PFA}}
\newcommand{\PBA}                            {{\rm PBA}}
\newcommand{\PRA}                            {{\rm PRA}}
\newcommand{\RPRA}                            {{\rm RatPRA}}

\newcommand{\HPBA}                            {{\rm HPBA}}
\newcommand{\RPBA}                            {{\rm RatPBA}}
\newcommand{\RHPBA}                            {{\rm RatHPBA}}

\newcommand{\rk}                           {\mathsf{rk}}
\newcommand{\intersect}                     {\mathbin{\cap}}
\newcommand{\union}                         {\mathbin{\cup}}

\newcommand{\murej}[2]                       {\mu^{rej}_{{#1},\, {#2}}}
\newcommand{\muacc}[2]                       {\mu^{acc}_{{#1},\,{#2}}}
\newcommand{\word}                                   {\alpha}
\newcommand{\Sigmaw}                      {\Sigma^\omega}
\newcommand{\cu}                                  {\eta}
\newcommand{\C}                                {\mathsf{C}}

\newcommand{\post}                            {\mathsf{post}}
\newcommand{\q}                                 {q}

\newcommand{\Nats}                                {\mathbb{N}}

\newcommand{\cR}                               {\mathcal{R}}

\newcommand{\cL}                               {\mathcal{L}}
\newcommand{\sL}                               {\mathsf{L}}

\newcommand{\cB}                               {\mathcal{B}}

\newcommand{\cF}                             {\mathcal{F}}
\newcommand{\cG}                             {\mathcal{G}}
\newcommand{\cal}[1]                             {\mathcal{#1}}

\newcommand{\cA}{{\cal A}}
\newcommand{\Cc}{{\cal C}}

\newcommand{\cl}[1]                           {\mathsf{cl}(#1)}
\newcommand{\half}                           {\frac{1}{2}}

\newcommand{\zero}                           {\mathbf{0}}
\newcommand{\one}                           {\mathbf{1}}
\newcommand{\val} [1]                        {\mathsf{bin}(#1)}

\newcommand{\Id}                                    {\mathsf{Id}}

\newcommand{\Regular}                         {\mathsf{Regular}}

\newcommand{\Pref}{\mathit{Pref}}
\newcommand{\pspace} {{\bf PSPACE}}

\newcommand {\cDD}{{\cal D}}
\newcommand {\cPP}{{\cal P}}

\newcommand{\re} {{\bf R.E.}}
\newcommand{\core} {{\bf co-R.E.}}
\newcommand{\nl} {{\bf NL}}

\newcommand{\exptime} {{\bf EXPTIME}}
\newcommand{\np} {{\bf NP}}

\newcommand{\ie} {{\it i.e.}}
\newcommand{\cT} {\mathbb{T}}
\newcommand{\new} {\mathsf{new}}

\newcommand{\powerset}[1] {2^ #1}

\newcommand{\I} {\mathcal{I}}

\newcommand{\card}[1] {\vert #1 \vert}
\newcommand{\set}[1] {\ensuremath\{#1 \}}

\newcommand{\SeqComp} {\mathsf{SeqComp}}

\newcommand{\comb}[2] {\begin{pmatrix} #1 \\ #2 \end{pmatrix}}

\tikzstyle{autst}=[draw,circle,minimum size=0.7cm]
\tikzstyle{finalst}=[draw,circle,double,minimum size=0.7cm]

\begin{document}
\title{Power of Randomization in Automata on Infinite Strings\rsuper*}

\author[R.~Chadha]{Rohit Chadha\rsuper a}	%required
\address{{\lsuper a}LSV, ENS Cachan \& CNRS \& INRIA, France}
\thanks{{\lsuper a}Supported in part by NSF grants CCF04-29639 and NSF CCF04-48178.}	%optional
\email{chadha.rohit@gmail.com }

\author[A.~P.~Sistla]{A.~Prasad Sistla\rsuper b}%optional
\address{{\lsuper b}Univ. of Illinois at Chicago, U.S.A.}	%optional
\email{sistla@cs.uic.edu }
\thanks{{\lsuper b}Supported in part by NSF CCF-0742686.}	%optional

\author[M.~Viswanathan]{Mahesh Viswanathan\rsuper c}	%optional
\address{{\lsuper c}Univ. of Illinois at Urbana-Champaign, U.S.A}	%optional
\email{vmahesh@cs.uiuc.edu }
\thanks{{\lsuper c}Supported in part by NSF CCF04-48178 and NSF CCF05-09321.}

%\maketitle
%\thispagestyle{empty}
\date{}

\begin{abstract}
 Probabilistic B\"{u}chi Automata (\PBA) are randomized, finite state
 automata that process input strings of infinite length. Based on the
 threshold chosen for the acceptance probability, different classes of
 languages can be defined. In this paper, we present a number of
 results that clarify the power of such machines and properties of the
 languages they define. The broad themes we focus on are as follows.
 We present results on the decidability and precise complexity of the
 emptiness, universality and language containment problems for such
 machines, thus answering questions central to the use of these models
 in formal verification.  Next, we characterize the languages
 recognized by {\PBA}s topologically, demonstrating that though
 general {\PBA}s can recognize languages that are not regular,
 topologically the languages are as simple as $\omega$-regular
 languages. Finally, we introduce Hierarchical {\PBA}s, which are
 syntactically restricted forms of {\PBA}s that are tractable and
 capture exactly the class of $\omega$-regular languages.
\end{abstract}

%\newpage
\subjclass{F.4.3,D.2.4,F.1.1,F.1.2}
\titlecomment{{\lsuper*}An extended abstract of the paper appeared in
\cite{rch:sis:mv:09}.}
\keywords{ Automata on infinite strings, Randomization, Omega-regular languages,
 Expressiveness, Decidability, Probabilistic Monitors}
%{Theory of Computation}{Formal Languages} 
%\category{D.2.4}{Software Engineering}{Program Verification}
%\category{F.1.1}{Theory of Computation}{Models of Computation}
%\category{F.1.2}{Theory of Computation}{Modes of Computation}

\maketitle

\section{Introduction}
Automata on infinite (length) strings have played a central role in
the specification, modeling and verification of non-terminating,
reactive and concurrent
systems~\cite{var:wol,kur,var:wol:sis,hol:pel,SistlaThesis}.  However,
there are classes of systems whose behavior is probabilistic in
nature; the probabilistic behavior being either due to the employment
of randomization in the algorithms executed by the system or due to
other uncertainties in the system, such as failures, that are modeled
probabilistically. While Markov Chains and Markov Decision Processes
have been used to model such behavior in the formal verification
community~\cite{marta-book}, both these models do not adequately
capture \emph{open}, \emph{reactive} probabilistic systems that
continuously accept inputs from an environment. The most appropriate
model for such systems are probabilistic automata on infinite strings,
which are the focus of study in this paper.

Probabilistic B\"uchi Automata (\PBA) have been introduced
in~\cite{Baier:lics2005} to capture such com\-putational devices. These
automata generalize probabilistic finite automata
(\PFA) \cite{Rabin,SALOMAA,PAZ} from finite length inputs to infinite
length inputs. Informally, {\PBA}s are like finite-state automata
except that they differ in two respects. First, from each state and on
each input symbol, the {\PBA} may roll a dice to determine the next
state. Second, the notion of acceptance is different because {\PBA}s
are probabilistic in nature and have infinite length input strings.
The behavior of a {\PBA} on a given infinite input string can be
captured by an infinite Markov chain that defines a probability
measure on the space of runs/executions of the machine on the given
input. Like B\"{u}chi automata, a run is considered to be accepting if
some accepting state occurs infinitely often, and therefore, the
probability of acceptance of the input is defined to be the measure of
all accepting runs on the given input. There are two possible
languages that one can associate with a {\PBA}
$\cB$~\cite{Baier:lics2005,Baier:fossacs2008} --- $\cL_{>0}(\cB)$
(called \emph{probable semantics}) consisting of all strings whose
probability of acceptance is non-zero, and $\cL_{=1}(\cB)$ (called
\emph{almost sure semantics}) consisting of all strings whose probability
of acceptance is $1$. Based on these two languages, one can define two
classes of languages --- $\probable$, and $\alsure$ which are the
collection of all languages (of infinite length strings) that can be
accepted by some {\PBA} with respect to probable, and almost sure
semantics, respectively. In this paper we study the expressive power
of, and decision problems for these classes of languages.

We present a number of new results that highlight three broad themes.
First, we establish results on decidability and  precise complexity of the canonical decision
problems in verification, namely, emptiness, universality, and
language containment, for the classes $\probable$ and $\alsure$.  For
the decision problems, we focus our attention on {\RPBA}s which are
{\PBA}s in which all transition probabilities are rational. For
{\RPBA}s $\cB$ and $\cB'$, our results are as follows.
\begin{enumerate}[(A)]
\item Checking if $\cL_{=1}(\cB) = \emptyset$ and $\cL_{=1}(\cB)
  = \Sigma^*$ are \pspace-complete.
\item The problems of checking if $\cL_{>0}(\cB) = \emptyset$,
  and if $\cL_{>0}(\cB) = \Sigma^*$ are $\mathbf{\Sigma}^0_2$-complete.
\item The problems of checking if $\cL_{=1}(\cB) \subseteq
  \cL_{=1}(\cB')$ and if $\cL_{>0}(\cB) \subseteq \cL_{>0}(\cB')$ are
  $\mathbf{\Sigma}^0_2$-complete.
\end{enumerate}
The decidability of the universality
checking of $\cL_{=1}(\cB)$ (bullet (A) above) is a new result. 
The result establishing the \pspace-completeness of emptiness checking
of $\cL_{=1}(\cB)$ (bullet (A) above) substantially improves the
result of~\cite{Baier:fossacs2008} where it was shown to be decidable
in {\exptime} and conjectured to be \exptime-hard. The improved upper
bound for emptiness checking is established by observing that the complement of the language
$\cL_{=1}(\cB)$ is recognized by a special {\PBA} $\cM$ (with probable
semantics) called a \emph{finite state probabilistic monitor}
(\FPM)~\cite{rch:sis:mv:08,rch:sis:mv:08a} and then exploiting a
result in~\cite{rch:sis:mv:08a} that shows that the language of an
{\FPM} is universal if and only if there is an \emph{ultimately
  periodic word} in the complement of the language recognized by a {\FPM}. This observation of the existence of
ultimately periodic words does not carry over to the class
$\probable$. However, we show that $\cL_{>0}(\cB)$ is non-empty iff it
contains a \emph{strongly asymptotic word}, which is a generalization
of ultimately periodic word. This allows us to show that the emptiness
problem for $\probable$, though undecidable as originally shown
in~\cite{Baier:fossacs2008}, is $\mathbf{\Sigma}^0_2$-complete (bullet (B)
above), where $\mathbf{\Sigma}^0_2$ is a set in the second level of the
arithmetic hierarchy.  This result is noteworthy because typically
problems of automata on infinite words that are undecidable tend to
lie way beyond the arithmetic hierarchy in the analytical
hierarchy. Finally, given that the emptiness and universality problems
for $\alsure$ are in {\pspace} (bullet (A)), one would expect language
containment under almost sure semantics to be at least decidable. However, surprisingly, we show
that it is, in fact, $\mathbf{\Sigma}^0_2$-complete (bullet (C) above).

The second theme brings to sharper focus the correspondence between
nondeterminism and probable semantics, and between determinism and
almost sure semantics, in the context of automata on infinite words.
This correspondence was hinted at in~\cite{Baier:fossacs2008}. There
it was observed that $\alsure$ is a strict subset of $\probable$ and
that while B\"{u}chi, Rabin and Streett acceptance conditions all
yield the same class of languages under the probable semantics, they
yield different classes of languages under the almost sure semantics.
These observations mirror the situation in non-probabilistic automata
--- languages recognized by deterministic B\"{u}chi automata are a
strict subset of the class of languages recognized by nondeterministic
B\"{u}chi automata, and while B\"{u}chi, Rabin and Streett acceptances
are equivalent for nondeterministic machines, B\"{u}chi acceptance is
strictly weaker than Rabin and Streett for deterministic machines. In
this paper we further strengthen this correspondence through a number
of results on the closure properties as well as the topological
structure of $\probable$ and $\alsure$.

First we consider closure properties. It was shown
in~\cite{Baier:fossacs2008} that the class $\probable$ is closed under
all the Boolean operations (like the class of languages recognized by
nondeterministic B\"{u}chi automata) and that $\alsure$ is not closed
under complementation. We extend these observations as follows.
\begin{iteMize}{(A)}
\item[(D)] $\alsure$ is closed under intersection and union.
\item[(E)] Every language in $\probable$ can be expressed as the
  Boolean combination of languages in $\alsure$.
\end{iteMize}
These results mimic similar observations about B\"{u}chi automata ---
the class of languages recognized by deterministic B\"{u}chi automata
is closed under union and intersection, but not complementation; and,
any $\omega$-regular language (or languages recognized by
nondeterministic B\"{u}chi machines) can be expressed as the Boolean
combination of languages recognized by deterministic B\"{u}chi
automata.

Next, we characterize the classes topologically. There is a natural
topological space on infinite length strings called the \emph{Cantor
  topology}~\cite{thomas-handbook}. We show that, like
$\omega$-regular languages, all the classes of languages defined by
{\PBA}s lie in very low levels of this Borel hierarchy.  We show that--
\begin{iteMize}{(A)}
\item[(F) ]$\alsure$ is strictly contained in $\cG_\delta$, just like the class
of languages recognized by deterministic B\"{u}chi is strictly
contained in $\cG_\delta$. 
\item[(G)]
$\probable$ is strictly contained in the Boolean closure of $\cG_\delta$ much like the
case for $\omega$-regular languages.\smallskip
\end{iteMize}

\noindent The last theme identifies syntactic restrictions on {\PBA}s that
captures regularity. Much like {\PFA}s for finite word languages,
{\PBA}s, though finite state, allow one to recognize non-regular
languages. It has been shown~\cite{Baier:lics2005,Baier:fossacs2008}
that both $\probable$ and $\alsure$ contain non-$\omega$-regular
languages. A question initiated in~\cite{Baier:lics2005} was to
identify restrictions on {\PBA}s that ensure that {\PBA}s have the
same expressive power as finite-state (non-probabilistic) machines.
One such restriction was identified in~\cite{Baier:lics2005}, where it
was shown that \emph{uniform} {\PBA}s with respect to the probable
semantics capture exactly the class of $\omega$-regular languages.
However, the uniformity condition identified by Baier et. al. was
semantic in nature. In this paper, we identify one simple syntactic
restriction (i.e., one that is based only on the local transition
structure of the machine, and can be efficiently checked) that
captures regularity both for probable semantics and almost sure
semantics. Not only, the restricted {\PBA}s capture  the notion of 
regularity, they are also very tractable. 

The restriction we consider is that of a hierarchical
structure. A \emph{Hierarchical {\PBA} (\HPBA)} is a {\PBA} whose
states are partitioned into different levels such that, from any state
$q$, on an input symbol $a$, at most one transition with non-zero
probability goes to a state at the same level as $q$ and all others go
to states at higher level. We show that --
\begin{iteMize}{(A)}
\item[(H)] {\HPBA}s with respect to
probable semantics define exactly the class of $\omega$-regular
languages.
\item[(I)] {\HPBA}s with respect to almost sure semantics define exactly
the class of $\omega$-regular languages in $\alsure$, namely, those
recognized by deterministic B\"{u}chi automata.
\item[(J)] Emptiness and universality
problems for probable semantics for {\HPBA}s with rational transition probabilities are \nl-complete
and \pspace-complete, respectively.
\item[(K)] Emptiness and universality
problems for 
almost sure semantics for {\HPBA}s with rational transition probabilities 
 \pspace-complete and  \nl-complete,
respectively.
\end{iteMize}
% Though {\HPBA}s have
%the same expressive power as (non-probabilistic) finite state
%automata, they can be exponentially more succinct, i.e., there is a
%family of languages such that the smallest {\HPBA} is exponentially
%smaller than the smallest deterministic machine.
The complexity of decision problems for {\HPBA}s under probable semantics
is interesting because this is the exact same complexity as that for
(non-probabilistic) B\"{u}chi automata. In contrast, the emptiness
problem for uniform {\PBA} has been shown to be in {\exptime} and
co-\np-hard~\cite{Baier:lics2005}; thus, they seem to be less
tractable than {\HPBA}.

The rest of the paper is organized as follows. After discussing
closely related work, we start with some preliminaries (in
Section~\ref{sec:prelim}) before introducing {\PBA}s. We present our
results about the probable semantics in Section~\ref{sec:probable},
and almost sure semantics in Section~\ref{sec:alsure}. Hierarchical
{\PBA}s are introduced in Section~\ref{sec:hier}, and conclusions are
presented in Section~\ref{sec:conc}.

{\bf Related Work.}  Probabilistic B\"{u}chi automata (\PBA),
introduced in~\cite{Baier:lics2005}, generalize the model of
Probabilistic Finite Automata~\cite{Rabin,SALOMAA,PAZ} to consider
inputs of infinite length. In~\cite{Baier:lics2005}, Baier and
Gr\"{o}{\ss}er only considered the probable semantics for {\PBA}. They
also introduced the model of uniform {\PBA}s to capture
$\omega$-regular languages and showed that the emptiness problem for
such machines is in {\exptime} and co-\np-hard. The almost sure
semantics for {\PBA} was first considered in~\cite{Baier:fossacs2008}
where a number of results were established. It was shown that
$\probable$ are closed under all Boolean operations, $\alsure$ is
strictly contained in $\probable$, the emptiness problem for
$\probable$ is undecidable, and the emptiness problem of $\alsure$ is
in {\exptime}. We extend and sharpen the results of this paper. In a
series of previous papers~\cite{rch:sis:mv:08,rch:sis:mv:08a}, we
considered a special class of {\PBA}s called {\FPM}s (Finite state
Probabilistic Monitors) whose accepting set of states consists of all
states excepting a rejecting state which is also absorbing.  There we
proved a number of results on the expressiveness and
decidability/complexity of problems for {\FPM}s. We draw on many of
these observations to establish new results for the more general model
of {\PBA}s. 

An extended abstract of this paper appeared in
\cite{rch:sis:mv:09}. Several proofs were omitted in
\cite{rch:sis:mv:09} for lack of space, and the current version
includes all of these proofs.

%%%%%%%%%%%%%%%%%%%%%%%%%%%%%%%%%%%%%%%%%%%%%%%%%%%%%%%%%%%%%%%%%%%%%

%%%%%%%%%%%%%%%%%%%%%%%%%%%%%%%%%%%%%%%%%%%%%%%%%%%%%%%%%%%%%%%%%%%%%%%%

\section{Preliminaries}
\label{sec:prelim}

%We assume that the reader is familiar with arithmetical
%hierarchy
%for the sake of convenience of the reader, we have also introduced them 
%in Appendix \ref{app:prelim}).  
The set of natural numbers will be denoted by $\Nats$,
the closed unit interval by $[0,1]$ and the open unit interval by
$(0,1).$ The power-set of a set $X$ will be denoted by $\powerset X.$

\vspace*{0.1cm}
\noindent {\bf Sequences.}  Given a finite set $S$, $\card S$ denotes
the cardinality of $S$.  Given a sequence (finite or infinite)
$\kappa=s_0,s_1,\ldots$ over $S$, $\card\kappa$ will denote the length
of the sequence (for infinite sequence $\card \kappa$ will be
$\omega$), and $\kappa[i]$ will denote the $i$th element $s_i$ of the
sequence. As usual $S^*$ will denote the set of all finite
sequences/strings/words over $S$, $S^+$ will denote the set of all finite
non-empty sequences/strings/words over $S$ and $S^\omega$ will denote the set of
all infinite sequences/strings/words over $S$. Given $\eta\in S^*$ and
$\kappa \in S^*\union S^\omega$, $\eta\kappa$ is the sequence obtained
by concatenating the two sequences in order.  Given $\sL_1\subseteq
\Sigma^*$ and $\sL_2\subseteq \Sigmaw$, the set $\sL_1\sL_2$ is
defined to be $\set{\eta\kappa\st \eta\in \sL_1 \textrm{ and } \kappa
  \in\sL_2}.$ Given natural numbers $i,j\leq \card \kappa$,
$\kappa[i:j]$ is the finite sequence $s_i,\ldots s_j$, where $s_k =
\kappa[k]$.  The set of {\it finite prefixes} of $\kappa$ is the set
$\Pref{(\kappa)}=\set{\kappa[0,j]\st j \in \Nats, j\leq\card \kappa}$.

\vspace*{0.1cm}
\noindent{\bf Arithmetical Hierarchy.}
Let $\Gamma$ be a finite alphabet. A language $\sL$ over $\Gamma$ is a set of 
finite strings over $\Gamma$. Arithmetical hierarchy consists of 
classes of languages $\mathbf{\Sigma}^0_n,\;\mathbf{\Pi}^0_n$ for each integer $n>0$.
Fix an $n>0$. A language $\sL\in \mathbf{\Sigma}^0_n$ iff there exists a recursive
predicate $\phi(u,\vec{x}_1,...,\vec{x}_n)$ where $u$ is a variable ranging
over $\Gamma^*$, and for each $i$,$0<i\leq n$,
$\vec{x}_i$ is a finite sequence of variables
ranging over integers such that 
$$L\;=\;\{u\in \Gamma^*\st \exists \vec{x}_1,\forall \vec{x}_2,\ldots,
Q_n \vec{x}_n \;\phi(u,\vec{x}_1,...,\vec{x}_n)\}$$ where $Q_n$ is an
existential quantifier if $n$ is odd, else it is a universal
quantifier. Note that the quantifiers in the above equation are
alternating starting with an existential quantifier. The class
$\mathbf{\Pi}^0_n$ is exactly the class of languages that are
complements of languages in
$\mathbf{\Sigma}^0_n$. $\mathbf{\Sigma}^0_1,\;\mathbf{\Pi}^0_1$ are
exactly the class of \re-sets and \core-sets. A canonical
$\mathbf{\Sigma}^0_1$-complete~\footnote{Let ${\cal C}$ be a class in
  the arithmetic hierarchy. $L \in {\cal C}$ is said to be ${\cal
    C}$-complete if $L \in C$, and for every $L' \in {\cal C}$ there
  is a computable function $f$ such that $x \in L'$ iff $f(x) \in L$.}
language is the set of deterministic Turing machine encodings that
halt on some input string.  A well known
$\mathbf{\Sigma}^0_2$-complete language is the set of deterministic
Turing machine encodings that halt on finitely many inputs.

\vspace*{0.1cm}
\noindent {\bf Languages of infinite words.}  A language $\sL$ of
infinite words over a finite alphabet $\Sigma$ is a subset of
$\Sigmaw.$ (Please note we restrict only to finite alphabets).  A set
of languages of infinite words over $\Sigma $ is said to be a class of
languages of infinite words over $\Sigma$.  Given a class $\cL$, the
Boolean closure of $\cL$, denoted $\cBool(\cL)$, is the smallest class
containing $\cL$ that is closed under the Boolean operations of
complementation, union and intersection.

\vspace*{0.1cm}
\noindent{\bf Automata and $\omega$-regular Languages.}  A
\emph{finite automaton on infinite words}, ${\cal A}$, over a (finite)
alphabet $\Sigma$ is a tuple $(\Q,q_0,F,\Delta)$, where $\Q$ is a
finite set of states, $\Delta\subset \Q\times\Sigma\times \Q $ is the
transition relation, $q_0\in \Q$ is the initial state, and $F$ defines
the accepting condition. The nature of $F$ depends on the type of
automaton we are considering; for a \emph{B\"{u}chi automaton} $F
\subseteq \Q$, while for a \emph{Rabin automaton} $F$ is a finite
subset of $\powerset{\Q} \times \powerset{Q}$. If for every $q\in Q$
and $a\in \Sigma$, there is exactly one $q'$ such that $(q,a,q')\in
\Delta$ then ${\cal A}$ is called a {\it deterministic} automaton.
Let $\alpha\,=\, a_0, a_1,\ldots$ be an infinite string over
$\Sigma$. A {\it run $r$ of ${\cal A}$} on $\alpha$ is an infinite
sequence $s_0,s_1,\ldots$ over $Q$ such that $s_0 =q_0$ and for every
$i\geq 0$, $(s_{i},a_i,s_{i+1})\in \Delta$. The notion of an
\emph{accepting run} depends on the type of automaton we consider. For
a B\"{u}chi automaton, $r$ is accepting if some state in $F$ appears
infinitely often in $r$. On the other hand for a Rabin automaton, $r$
is accepting if it satisfies the {\it Rabin acceptance condition} ---
there is some pair $(B_i,G_i) \in F$ such that all the states in $B_i$
appear only finitely many times in $r$, while at least one state in
$G_i$ appears infinitely many times.  The automaton ${\cal A}$ {\it
  accepts} the string $\alpha$ if it has an accepting run on $\alpha$.
The {\it language accepted (recognized) by ${\cal A}$}, denoted by
$\cL({\cal A})$, is the set of strings that ${\cal A}$ accepts. A
language $\sL \subseteq \Sigmaw$ is called \emph{$\omega$-regular} iff
there is some B\"{u}chi automata ${\cal A}$ such that $\cL({\cal A}) =
\sL$. In this paper, given a fixed alphabet $\Sigma$, we will denote
the class of $\omega$-regular languages by $\Regular$. It is
well-known that unlike the case of finite automata on finite strings,
deterministic B\"{u}chi automata are less powerful than
nondeterministic B\"{u}chi automata. On the other hand,
nondeterministic Rabin automata and deterministic Rabin automata have
the expressive power and they recognize exactly the class
$\Regular$. Finally, we will sometimes find it convenient to consider
automata ${\cal A}$ that do not have finitely many states. We will say
that a language $\sL$ is \emph{deterministic} iff it can be accepted
by a deterministic B\"{u}chi automaton that does not necessarily have
finitely many states. We denote by $\Det$ the collection of all
deterministic languages. Please note that the class $\Det$ strictly
contains the class of languages recognized by finite state
deterministic B\"uchi automata. The following are well-known
results~\cite{book,thomas-handbook}.
\begin{prop}\rm
\label{prop:reg-det}
 $\sL \in \Regular\intersect \Det$ iff there is a finite state
deterministic B\"uchi automaton $\cA$ such that $\cL(\cA)=\sL.$
Furthermore, $\Regular\intersect \Det \subsetneq \Regular$ and
$\Regular = \cBool(\Regular\intersect \Det)$.
\end{prop}

\vspace*{0.1cm}
\noindent{\bf Topology on infinite strings.} The set $\Sigmaw$
comes equipped with a natural topology called the \emph{Cantor
  topology}. The collection of open sets is the collection
$\cG=\set{\sL\Sigmaw\st \sL\subseteq \Sigma^+}$.\footnote{This
  topology is also generated by the 
  metric $d:\Sigmaw\times \Sigmaw \to [0,1]$ where $d(\alpha,\beta)$
  is $0$ iff $\alpha=\beta$; otherwise it is $\frac{1}{2^i} $ where
  $i$ is the smallest integer such that $\alpha[i]\ne \beta[i].$ } The
collection of closed sets, $\cF$, is the collection of {\it
  prefix-closed sets} --- ${\sL}$ is prefix-closed if for every
infinite string ${\alpha}$, if every prefix of ${\alpha}$ is a prefix
of some string in $\sL$, then ${\alpha}$ itself is in $\sL$. In the
context of verification of reactive systems, closed sets are also
called {\it safety languages}~\cite{La85,AS85}. 

\begin{figure}
\begin{center}
\begin{tikzpicture}
\input{borel-hierarchy.tkz}
\end{tikzpicture}
\end{center}
\caption{The Borel Hierarchy. Inclusions from left to right are strict.}
\label{fig:borel}
\end{figure}

\vspace*{0.1cm}
\noindent{\bf Borel Hierarchy on the Cantor space}.  For a class $\cL$
of languages, we define $\cL_\delta=\set{\intersect_{i\in \Nats}
  \sL_i\st \sL_i\in\cL }$ and $\cL_\sigma=\set{\union_{i\in \Nats}
  \sL_i\st \sL_i\in\cL }$.  The set of open sets of the Cantor space
is closed under arbitrary unions but only finite
intersections. Similarly the set of closed sets of the Cantor union is
closed arbitrary intersections but only finite unions. The Borel
hierarchy of the Cantor space is obtained by the means of countable
unions, intersections and complementation, and is shown in
Figure~\ref{fig:borel}.  This yields a transfinite hierarchy, but we
will restrict our attention to the first few levels.  At the lowest
level of this hierarchy is the collection $\cG\intersect \cF$ which is
strictly contained in both $\cG$ and $\cF$ which form the next level
of the hierarchy.  Both $\cG$ and $\cF$ are strictly contained in the
collection $\cG_{\delta} \intersect \cF_{\sigma}$ which forms the next
level. The collection $\cG_{\delta}\intersect \cF_{\sigma}$ is
strictly contained in $\cG_{\delta}$ and $\cF_{\sigma}$ which is at
the next level.  $\cG_{\delta}$ and $\cF_{\sigma}$ are strictly
contained in $\cG_{\delta\sigma}\intersect \cF_{\sigma\delta}$ which
itself is strictly contained in both $\cG_{\delta\sigma}$ and
$\cF_{\sigma\delta}$.  One remarkable result in automata theory is
that the class of languages $\cG_\delta$ coincides exactly with the
class of languages recognized by infinite-state deterministic B\"uchi
automata \cite{Landweber69,book,thomas-handbook}. This combined with the fact that
the class of $\omega$-regular languages is the Boolean closure of
$\omega$-regular deterministic B\"uchi automata yields that the class
of $\omega$-regular languages is strictly contained in
$\cBool(\cG_\delta)$ which itself is strictly contained in
$\cG_{\delta\sigma}\intersect \cF_{\sigma\delta}$
\cite{book,thomas-handbook}.

\begin{prop}\rm
$\cG_{\delta}=\Det,$ and $\Regular \subsetneq \cBool(\cG_\delta)
  \subsetneq \cG_{\delta\sigma}\intersect \cF_{\sigma\delta}.$
\end{prop}

\subsection{Probabilistic B\"uchi automata}
\label{sec:pba}
We shall now recall the definition of probabilistic B\"uchi automata
given in \cite{Baier:lics2005}.  Informally, {\PBA}s are like
finite-state deterministic B\"uchi automata except that the transition
function from a state on a given input is described as a probability
distribution that determines the probability of the next state.
{\PBA}s generalize the probabilistic finite automata
(\PFA)~\cite{Rabin,SALOMAA,PAZ} on finite input strings to infinite
input strings.  Formally,
\begin{defi}\rm
A {\it finite state probabilistic B\"uchi automata ({\PBA})} over a
finite alphabet $\Sigma$ is a tuple $\cB=(\Q,\qs,\Qf,\delta )$ where
$\Q$ is a finite set of {\it states}, $\qs \in \Q$ is the {\it initial
  state}, $\Qf \subseteq \Q$ is the set of {\it accepting/final states}, and
$\delta:\Q\times \Sigma \times \Q\to \unit$ is the {\it transition
  relation} such that for all $q\in \Q$ and $a \in \Sigma$, $
\sum_{q'\in \Q} \delta(q,a,q') =1.$ In addition, if $\delta(q,a,q')$
is a rational number for all $q,q'\in \Q, a \in \Sigma$, then we say
that $\cM$ is a rational probabilistic B\"uchi automata (\RPBA).
\end{defi}

\begin{notation}
The transition function $\delta$ of {\PBA} $\cB$ on input $a$ can be
seen as a square matrix $\delta_a$ of order $\card Q$ with the rows
labeled by ``current'' state, columns labeled by ``next state'' and
the entry $\delta_a(\q,\q')$ equal to $ \delta(\q,a,\q')$. Given a
word $u = a_0a_1\ldots a_n \in \Sigma^+$, $\delta_u$ is the matrix
product $\delta_{a_0}\delta_{a_1}\ldots \delta_{a_n}.$ For an empty
word $\epsilon\in \Sigma^*$ we take $\delta_\epsilon$ to be the
identity matrix.  Finally for any $\Q_0\subseteq \Q,$ we define
$\delta_u(q,\Q_0)=\sum_{q'\in \Q_0}\delta_u(q,q').$ Given a state
$q\in Q$ and a word $u\in \Sigma^+$, $\post(q,u)=\set{q'\st
  \delta_u(q,q')>0}.$
\end{notation}

Intuitively, the {\PBA} starts in the initial state $\qs$ and if after
reading $a_0,a_1\ldots,a_i$ results in state $\q$, then it moves to
state $\q'$ with probability $\delta_{a_{i+1}}(\q,\q')$ on symbol
$a_{i+1}$.  Given a word $\word\in \Sigmaw$, the \PBA\ $\cB$ can be
thought of as a infinite state Markov chain which gives rise to the
standard $\sigma$-algebra defined using cylinders and the standard
probability measure on Markov chains~\cite{markov1,markov2} as
follows.
%Given a \PBA, $\cB=\fpm{\Q,\qs,\Qf,\delta }$ on the alphabet $\Sigma$
%and a word $\word \in \Sigmaw$, 
Given a word $\word \in \Sigmaw$, the {\em probability space
  generated} by $\cB$ and $\word$ is the probability space
$(\Q^\omega,\mathcal{F}_{\cB,\word}, \mu_{\cB,\word}) $ where
\begin{iteMize}{$\bullet$}
\item $\mathcal{F}_{\cB,\word}$ is the smallest $\sigma$-algebra on
  $\Q^\omega$ generated by the collection $\{\C_\cu\st \cu \in \Q^+\}$
  where $\C_\cu=\{\rho \in \Q^\omega \st \cu \textrm{ is a prefix of }
  \rho \}$.
         
\item $\mu_{\cB,\word}$ is the unique probability measure on
  $(\Q^\omega,\mathcal{F}_{\cB,\word})$ such that
  $\mu_{\cB,\word}(\C_{q_0\ldots q_n})$ is
\begin{iteMize}{$-$}
\item $0$ if $q_0 \ne q_s$, 
\item $1$ if $n=0$ and $q_0=\qs$, and 
\item $\delta(q_0,\word(0),q_1)\ldots \delta(q_{n-1},\word(n-1),q_n)$
  otherwise.
\end{iteMize}
\end{iteMize}\smallskip
%\end{defi} 

% We denote this measure by $\mu_{\cB,\alpha}$. 
\noindent A {\em run} of the \PBA\ $\cB$ is an infinite sequence $\rho \in
\Q^\omega$.  A run $\rho$ is {\em accepting} if $\rho[i]\in \Qf$ for
infinitely many $i$. A run $\rho$ is said to be {\em rejecting} if it
is not accepting.  The set of accepting runs and the set of rejecting
runs are measurable~\cite{markov1}. Given a word $\word$, the measure
of the set of accepting runs is said to be the {\it probability of
  accepting $\word$} and is henceforth denoted by $\muacc\cB\word$;
and the measure of the set of rejecting runs is said to be the {\it
  probability of rejecting $\word$} and is henceforth denoted by
$\murej\cB\word$.  Clearly $\muacc\cB\word+\murej\cB\word=1.$
Following, \cite{Baier:lics2005,Baier:fossacs2008}, a {\PBA} $\cB$ on
alphabet $\Sigma$ defines two {\it semantics}:
\begin{iteMize}{$\bullet$}
\item $\cL_{>0}(\cB)=\set{\alpha\in \Sigmaw\st\muacc\cB\word >0}$,
  henceforth referred to as the {\it probable semantics} of $\cB$, and
 
\item $\cL_{=1}(\cB)=\set{\alpha\in \Sigmaw\st \muacc\cB\word =1}$,
   henceforth referred  to as the {\it almost-sure semantics} of $\cB$.
   
\end{iteMize}
This gives rise to the following classes of languages of infinite
words.
\begin{defi}\rm
Given a finite alphabet $\Sigma$, 
\begin{iteMize}{$\bullet$}
\item 
$\probable=\set{\sL\subseteq \Sigmaw\st \exists{ \PBA}\; \cB.\; \sL=\cL_{> 0}(\cB) };$

\item
$\alsure=\set{\sL\subseteq \Sigmaw\st \exists{ \PBA}\; \cB.\; \sL=\cL_{=1 }(\cB)  }.$
\end{iteMize}
\end{defi}  

\noindent
{\bf Probabilistic Rabin automaton.} Analogous to the definition of a
{\PBA} and \RPBA, one can define a Probabilistic Rabin automaton
{\PRA} and {\RPRA} \cite{Baier:fossacs2008,GroesserThesis}; where
instead of using a set of final states, a set of pairs of subsets of
states is used. A run in that case is said to be accepting if it
satisfies the Rabin acceptance condition. It is shown in
\cite{Baier:fossacs2008,GroesserThesis} that {\PRA}s have the same
expressive power under both probable and almost-sure
semantics. Furthermore, it is shown in
\cite{Baier:fossacs2008,GroesserThesis} that for any {\PBA} $\cB$,
there is {\PRA} $\cR$ such that a word $\alpha$ is accepted by $\cR$
with probability $1$ iff $\alpha$ is accepted by $\cB$ with
probability $>0$. All other words are accepted with probability $0$ by
$\cR$.
\begin{prop}[\cite{Baier:fossacs2008}]\rm
\label{prop:buchi-rabin}
For any {\PBA} $\cB$ there is a {\PRA} $\cR$ such that
$\cL_{>0}(\cB)=\cL_{>0}(\cR)=\cL_{=1}(\cR)$ and
$\cL_{=0}(\cB)=\cL_{=0}(\cR).$ Furthermore, if $\cB$ is a $\RPBA$ then
$\cR$ is a $\RPRA$ and the construction of $\cR$ is recursive.
\end{prop}

%\vspace*{-0.1cm}
\noindent
{\bf Finite probabilistic monitors (\FPM)s. } We identify one useful
syntactic restriction of {\PBA}s, called \emph{finite probabilistic
  monitors} (\FPM)s.  In an {\FPM}, all the states are accepting
except a special absorbing \emph{reject} state. We studied them
extensively in \cite{rch:sis:mv:08,rch:sis:mv:08a}.
\begin{defi}\rm
A {\PBA} $\cM=(\Q,\qs,\Qf,\delta)$ on $\Sigma$ is said to be an {\FPM}
if there is a state $\qr\in \Q$ such that $\qr \ne \qs$,
$\Qf=\Q\setminus\set{\qr}$ and $\delta(\qr,a,\qr)=1$ for each $a \in
\Sigma.$ The state $\qr$ said to be the {\it reject} state of $\cM$. If in
addition $\cM$ is a $\RPBA$, we say that $\cM$ is a rational finite
probabilistic monitor ($\RFPM$).
\end{defi}

\section{Probable semantics}
\label{sec:probable}

In this section, we shall study the expressiveness of the languages
contained in $\probable$ as well as the complexity of deciding
emptiness and universality of $\cL_{>0}(\cB)$ for a given {\RPBA}
$\cB$.  We assume  that the alphabet $\Sigma$
is fixed and contains at least two letters.%We shall later restrict our attention to hierarchical {\PBA}s.

%\vspace*{-0.2cm}
\subsection{Expressiveness} 
We shall establish new expressiveness results for the class $\probable$--
\begin{iteMize}{$\bullet$}
\item 
          We  show that although the class $\probable$ strictly contains $\omega$-regular languages \cite{Baier:lics2005}, 
           it is not topologically harder. More precisely, we will show that for any {\PBA} $\cB$, 
           $\cL_{>0}(\cB)$ is a $\cBool(\cG_\delta)$-set. This will be  a consequence of following facts.
            \begin{enumerate}[(a)]
            \item $\probable=\cBool(\alsure)$ (see Theorem \ref{thm:closure}).
            \item $\alsure \subseteq \cG_\delta$ (see Lemma \ref {lem:Gdelta}).
           \end{enumerate}
\item However, there are $\cBool(\cG_\delta)$ sets that are not in $\probable$ (see Lemma \ref{lem:probableexp}).\smallskip
\end{iteMize}

\noindent Our expressiveness results are summarized in
Figure~\ref{fig:exp-results}.

\begin{figure}
\begin{center}
\begin{tikzpicture}
\input{topo-results.tkz}
\end{tikzpicture}
\end{center}
\caption{Relationship between languages recognized by {\PBA}s and sets
  in the Borel hierarchy defined by the Cantor topology. Arrows
  indicate strict containment. `Det. Reg.' refers to the class of
  languages recognized by deterministic B\"{u}chi automata, while
  `Regular' refers to the class of $\omega$-regular
  languages. Containment arrows with label $1$ were proved
  in~\cite{Baier:lics2005} and those labelled $2$ were proved
  in~\cite{Baier:fossacs2008}. Results relating the classes Regular
  and Det. Reg. are classical results; see
  survey~\cite{book,thomas-handbook}. Containment arrows  with label $3$
 and the
  equality $\cBool(\alsure)=\probable$  are proved in this paper.}
\label{fig:exp-results}
\end{figure}
 
We first show that just as the
class of $\omega$-regular languages is the Boolean closure of the
class of $\omega$-regular recognized by deterministic B\"uchi
automata,  the class $\probable$ coincides with the Boolean closure of
the class $\alsure.$ This is the content of the following theorem
whose proof is of independent interest and shall be used later in
establishing that the containment of languages of two PBAs under
almost-sure semantics is undecidable (see Theorem
\ref{thm:alcontundec}).
%The proof has been moved to Appendix \ref{app:pbaproofs} in the interests of
%space.
\begin{thm}\rm
\label{thm:closure}
$\probable=\cBool(\alsure).$
\end{thm}
\begin{proof}
First observe that it was already shown in \cite{Baier:fossacs2008}
that $\alsure\subseteq \probable.$ Since $\probable$ is closed under
Boolean operations, we get $\cBool(\alsure)\subseteq \probable.$
We have to show the reverse inclusion.

It suffices to show that given a {\PBA} $\cB$, the language
$\cL_{>0}(\cB)\in \cBool(\alsure).$ Fix $\cB.$ Recall that results of
\cite{Baier:fossacs2008,GroesserThesis} (see
Proposition~\ref{prop:buchi-rabin}) imply that there is a
probabilistic Rabin automaton ({\PRA}) $\cR$ such that 1)
$\cL_{>0}(\cB) = \cL_{=1} (\cR)=\cL_{>0}(\cR)$ and 2) $\cL_{=0}(\cB) =
\cL_{=0}(\cR)$.  Let $\cR=(\Q,\qs,F,\delta)$ where $F\subseteq
2^\Q\times 2^\Q$ is the set of the Rabin pairs. Assuming that $F$
consists of $n$-pairs, let $F=\set{(B_1,G_1),\ldots, (B_n,G_n)}$.

Given an index set $\I \subseteq \set{1,\ldots,n}$, let
$\Good_\I=\union_{r \in \I} G_r$. Let $\cR_\I$ be the {\PBA} obtained
from $\cR$ by taking the set of final states to be $\Good_\I$. In
other words, $\cR_\I=(\Q,\qs,\Good_\I,\delta).$ Given $\I \subseteq
\set{1,\ldots,n}$ and an index $j \in \I$, let $\Bad_{\I,j} =
B_j\union (\union_{r \in \I,r\ne j} G_r)$.  Let $\cR^j_\I$ be the {\PBA}
obtained from $\cR$ by taking the set of final states to be
$\Bad_{\I,j}$, {\it i.e.}, $\cR^j_\I=(\Q,\qs,\Bad_{\I,j},\delta).$ The
result follows from the following claim.
\begin{claim}
$$\cL_{>0}(\cB)= \bigcup_{\I\subseteq \set{1,\ldots,n}, j\in \I}
  \cL_{=1} (\cR_\I) \intersect (\Sigmaw\setminus
  \cL_{=1}(\cR^j_\I)).$$
\end{claim} 
\noindent {\bf Proof of the claim:} Given $\I\subseteq
\set{1,\ldots,n}, j\in \I$, let $\sL_{\I,j}=\cL_{=1} (\cR_\I)
\intersect (\Sigmaw\setminus \cL_{=1}(\cR^j_\I)).$ We will say that a
run $\rho$ of {\PRA} $\cR$ satisfies the Rabin pair $(B_r,G_r)$ if all
states in $B_r$ occur only finitely many times in $\rho$ and at least
one state in $G_r$ occurs infinitely often in $\rho.$

We first show that $\sL_{\I,j} \subseteq \cL_{>0}(\cR).$ Fix any
$\alpha\in\sL_{\I,j}.$ Since $\sL_{\I,j}\subseteq \cL_{=1}(\cR_\I)$,
it follows that on input $\alpha$ the measure of runs that visit the
set $\Good_\I= \union_{i\in \I}G_i$ infinitely often must be $1.$ On
the other hand, as $\sL_{\I,j} \cap \cL_{=1}(\cR^j_\I) = \emptyset$,
it follows that on input $\alpha$ the measure of runs that visit
$\Bad_{\I,j}=B_j\cup (\union_{i\in \I,i\ne j}G_i)$ only finitely many
times has strictly positive measure. Since $\Good_\I\setminus
\Bad_{\I,j}\subseteq G_j$, it now follows from the previous two
observations that the measure of runs that visit $G_j$ infinitely
often but visit $\Bad_{\I,j}$ only finitely many times is strictly
positive.  Since $B_j\subseteq \Bad_{\I,j}$, we get that the set of
runs that satisfy the Rabin pair $(B_j,G_j)$ has non-zero measure on
input $\alpha$.  Therefore, we have that $\sL_{\I,j}\subseteq
\cL_{>0}(\cR)$. But, we have that
$\cL_{>0}(\cR)=\cL_{=1}(\cR)=\cL_{>0}(\cB).$ Hence, we get
\[ 
\bigcup_{\I\subseteq \set{1,\ldots,n}, j\in \I} \sL_{\I,j}\subseteq \cL_{>0} (\cB). 
\]
We will be done if we can show the reverse inclusion.  Thus, given
word $\alpha$ in $\cL_{>0} (\cB)$, we have to construct $\I$ and $j$
such that $\alpha \in \sL_{I,j}$.  We construct them as follows.
First, let $\widetilde\I$ be the set of all indices $r$ such that the
measure of all runs that satisfy the Rabin pair $(B_r,G_r)$ on input
$\alpha$ is $>0.$ $\widetilde\I$ is non-empty (since $\alpha\in
\cL_{=1}(\cR)$).  Clearly, we have that on input $\alpha$, the measure
of runs such that $\Good_{\widetilde\I}$ is visited infinitely often
is $1$ (again, since $\alpha\in \cL_{=1}(\cR)$). In other words,
$\alpha\in\cL_{=1}(\cR_{\widetilde \I})$.  Required $\I$ will be a
subset of $\widetilde\I$ and will be constructed by induction as
follows.

At step $1$ of the induction, we pick an arbitrary index $r$ in
$\widetilde \I$. Then we check if it is the case that on $\alpha$, the
probability of visiting $G_r$ infinitely often in $\cR$ is $1$. Note 
that it is the case that the probability that $B_r$ is visited
infinitely often in $\cR$ is $<1$ (as $\alpha$ satisfies $(B_r,G_r)$
with non-zero probability). Note that this implies that $\alpha\in
\cL_{=1}(\cR_{\{r\}})\intersect (\Sigmaw\setminus
\cL_{=1}(\cR^r_{\{r\}})$ and the induction stops at this point. If it is not the
case, then let $\I_1=\set{r}$.

Proceed by induction. At step $m$, we would have produced an index set
$\I_m\subseteq \widetilde \I$ such that on $\alpha$, we have that
$\alpha \notin \cL_{=1} (\cR_{\I_m})$ (meaning the set of runs which
visit $\Good_{\I_m}$ infinitely often have probability $<1$). Now
since $\alpha$ is accepted by {\PRA} $\cR$ with probability $1$, there
must be some index $r$ in $\widetilde \I \setminus \I_m$ such that the
set of runs that satisfy $(B_r,G_r)$ and visit $\Good_{I_m}$ only
finitely many times is $>0$. Fix one such $r$.  Now, there are two
cases.
\begin{enumerate}[(1)]
\item On the input $\alpha$, the set of runs that visit
  $\Good_{\I_m}\cup G_r$ infinitely often has measure 1.  In that
  case, by construction, we also have that $\alpha\in \sL_{\I_m\cup
    \{r\},r}$ and induction stops.
\item Otherwise,  we let $\I_{m+1}= \I_m\cup\set{r}$ and proceed.
\end{enumerate}
The induction must stop at a finite point at which we will satisfy the
required condition (since $\alpha\in\cL_{=1}(\cR_{\widetilde \I})$).
%The proof of the claim is detailed in Appendix \ref{app:pbaproofs}.
%moved to Apendix \ref{app:pbaproofs} for lack of space reasons.
\end{proof}
The second component needed for showing that $\probable\subseteq
\cBool(\cG_\delta)$ is the fact that for any {\PBA} $\cB$ and $x\in
      [0,1]$, the language $\cL_{\geq x}(\cB)$ is a $\cG_\delta$-set;
      which we prove next.
\begin{lem}\rm
\label{lem:Gdelta}
For any {\PBA} $\cB$ and $x\in [0,1]$, $\cL_{\geq x}(\cB)$ is a $\cG_{\delta}$ set.
\end{lem}
\begin{proof}
Let $\cB=(\Q,\qs,\Qf,\delta).$ Now given $k>0$, let
$\Path^{k}\subseteq \Sigmaw$ be the set of all infinite runs which
start at the state $\qs$ and visit the set of final states at least
$k$-times. Let $\Path^\omega$ be the set of all infinite runs which
start at the state $\qs$ and visit the final states infinitely often.
Formally, $\Path^k=\set{\rho\in \Q^\omega\st \rho[0]=\qs \textrm{ and
  } \card{\set{i\in \Nats\st \rho[i]\in \Qf}}\geq k}$ and
$\Path^\omega=\set{\rho\in \Q^\omega\st \rho[0]=\qs \textrm{ and }
  \card{\set{i\in \Nats\st \rho[i]\in \Qf}}=\omega}.$ We have that
$Path^k, k>0$ forms a decreasing sequence and
$$\intersect_{k\in\Nats,k>0}\Path^k= \Path^\omega.$$
From standard probability theory, we get that for any word $\alpha,$
$$\lim_{k\to \infty} \mu_{\cB,\word}(\Path^k)=
\mu_{\cB,\word}(\Path^\omega)$$ where $\mu_{\cB,\word}$ is the
probability measure generated by the infinite word $\word$ and PBA
$\cB.$ From this, we immediately see that an infinite word $\alpha$ is
accepted with probability at least $x$ iff for all $k>0$ the
probability of visiting the set of final states on input $\alpha$ at
least $k$-times $\geq x$. In other words,
$$\set{\alpha\in \Sigmaw\st \muacc\cB\alpha \geq x} = \intersect_{k\in
  \Nats, k>0}\set{\alpha\in\Sigmaw\st \mu_{\cB,\word}(\Path^k) \geq x}
.$$ Hence, it suffices to show that for each $k\in \Nats,k>0$ the set
$\set{\alpha\in\Sigmaw\st \mu_{\cB,\word}(\Path^k) \geq x} $ is a
$\cG_\delta$ set. Note that for each $k>0$,
$$\set{\alpha\in\Sigmaw\st \mu_{\cB,\word}(\Path^k) \geq x}
=\intersect_{n\in \Nats} \set{\alpha\in\Sigmaw\st
  \mu_{\cB,\word}(\Path^k) > x-\frac{1}{n}} .$$ Hence, it suffices to
show that for each $k\in \Nats, n\in \Nats, k>0$ the set
$\set{\alpha\in\Sigmaw\st \mu_{\cB,\word}(\Path^k) > x-\frac{1}{n}} $
is a $\cG$-set. In order to see that this is the case, given $k>0$ and $\ell>0$, let $Path^{k,\ell}\subset \Sigmaw$
be the set of infinite runs that start at the initial state and visit $\Qf$ at least $k$ times in the first $\ell$ steps.
Formally, $\Path^{k,\ell}=\set{\rho\in \Q^\omega\st \rho[0]=\qs \textrm{ and
  } \card{\set{i\in \Nats, i <\ell\st \rho[i]\in \Qf}}\geq k}.$ 
  
  Now, the result  follows immediately from the observation that 
$$\set{\alpha\in\Sigmaw\st \mu_{\cB,\word}(\Path^k) > x-\frac{1}{n}}=\union_{\ell\in \Nats}\set{\alpha\in\Sigmaw\st \mu_{\cB,\word}(\Path^{k,\ell}) > x-\frac{1}{n}}$$
and the observation that each of the set $\set{\alpha\in\Sigmaw\st \mu_{\cB,\word}(\Path^{k,\ell}) > x-\frac{1}{n}}$ is a $\cG$-set.
\end{proof}

Using Lemma~\ref{lem:Gdelta}, one immediately gets that
$\probable\subseteq \cBool(\cG_\delta)$. Even though {\PBA}s accept
non-$\omega$-regular languages, they cannot accept all the languages
in $\cBool(\cG_\delta)$.
\begin{lem}\rm
\label{lem:probableexp}
$\Regular\subsetneq \probable\subsetneq \cBool(\cG_\delta).$
\end{lem}

\begin{proof}
Note that $\Regular\subsetneq \probable$ follows immediately from results of \cite{Baier:lics2005}.
Thanks to Lemma \ref{lem:Gdelta}, we also have that $\alsure \subseteq \cG_\delta.$
Since $\probable=\cBool(\alsure)$ (see Theorem \ref{thm:closure}), we
get that $\probable\subseteq \cBool(\cG_\delta).$ We only have to show
that this containment is strict. The proof of this fact utilizes the
following result which shows that for any $\sL\in \probable$, the
smallest safety language containing $\sL$ is guaranteed to be
$\omega$-regular even if $\sL$ is not.\footnote{As arbtitrary intersection of safety languages is also a safety language, for every language $\sL$, there is a smallest safety language containing $\sL$. Topologically, this is  the {\it closure} of $\sL$.}
\begin{claim}%\rm
%\label{lem:safety}
For any {\PBA} $\cB$, let $\cl{\sL}$ be the smallest safety language
containing $\sL=\cL_{>0}(\cB).$ Then $\cl{\sL}$ is $\omega$-regular.
\end{claim}
{\bf Proof of the claim:} Without loss of generality, we can assume
that $\sL\ne \emptyset.$ Let $\cB=(\Q,\qs, \Qf,\delta).$ Given $q\in
\Q$, let $\cB_q$ be the {\PBA} which is exactly like $\cB$, except
that the initial state is $q$. That is $\cB_q=(\Q,q, \Qf,\delta)$.
Let $\Q_{>0}\subseteq \Q$ be the set of states $\set{q\st \exists
  \alpha.  \muacc {\cB_q}\alpha >0}.$  
Consider the finite state
B\"uchi automata $\cA=(\Q_{>0},\qs,\Q_{>0},\Delta)$ where
$(q_1,a,q_2)\in \Delta$ iff $\delta(q_1,a,q_2)>0.$ It is easy to see
that $\cl\sL$ is exactly the language recognized by $\cA.$ This implies that $\cl\sL$ is $\omega$-regular.({\bf End
  of the claim}) \qed

We proceed as follows. Fix two letters $a,b$ of the alphabet $\Sigma$
and consider the language $\sL$ consisting of exactly one word
$\alpha=abaabb\ldots a^ib^ia^{i+1}b^{i+1}\ldots$. Now, $\cl\sL= \sL$
(every single element set in a metric space is a closed set) and
$\sL$ is not $\omega$-regular (since $\sL$ does not contain any
periodic word). Therefore, the closed set $\sL$ is not in the class
$\probable$ (note that $\sL\in\cG_\delta$ as $\cF\subseteq
\cG_\delta$).  
\end{proof}

%\begin{rems}
%Please note that Lemma \ref{lem:Gdelta} can be used to show that the classes $\strict$ and $\nstrict$
%are also contained within the first few levels of Borel hierarchy. However, we can show that no version of Theorem \ref{thm:closure} holds for those classes. More precisely, $\strict\not\subseteq \nstrict$
%and $\nstrict\not\subseteq \strict$. These results are out of the scope of this paper.
%\end{rems}

\subsection{Decision problems} 
For the rest of this section, we shall focus our attention on decision
problems for probable semantics for {\RPBA}s. Results of this section are summarized in the
 first row of Figure \ref{fig:compresults}  and stated in Theorem \ref{thm:pbaanalytical}.
Given a {\RPBA} $\cB$, the problems of emptiness and universality of
$\cL_{>0}(\cB)$ are known to be undecidable
\cite{Baier:fossacs2008}. We sharpen this result by showing that these
problems are $\mathbf{\Sigma}^0_2$-complete. This is interesting in the light
of the fact that problems on infinite string automata that are
undecidable tend to typically lie in the analytical hierarchy, and not
in the arithmetic hierarchy.

\begin{figure*}[t]
\footnotesize
\begin{center}
\begin{tabular}{| l | l | l | l |}

\hline
  & \multicolumn{1}{c|}{{Emptiness}} & \multicolumn{1}{c|}{{Universality}} &  \multicolumn{1}{c|}{{Containment}} \\

  \hline
  
$\probable$  &$\mathbf{\Sigma}^0_2$-complete$^{(\dagger)}$ & $\mathbf{\Sigma}^0_2$-complete$^{(\dagger)}$ & $\mathbf{\Sigma}^0_2$-complete$^{(\dagger)}$ \\

\hline
$\alsure$ &$\pspace$-complete$^{(\dagger\dagger)}$ & $\pspace$-complete & $\mathbf{\Sigma}^0_2$-complete \\
%$$ & ploynom& 
\hline
  \end{tabular}
\end{center}  

\caption{Hardness of decision problems for {\RPBA}s. $^{(\dagger)}$The problems of checking emptiness, universality and containment for probable semantics was shown to be $\re$-hard in  \cite{Baier:fossacs2008}.  $^{(\dagger\dagger)}$the problem of checking emptiness of
almost sure semantics was shown to decidable in $\exptime$ in  \cite{Baier:fossacs2008}.} 
\label{fig:compresults}
\end{figure*}

% Our upper-bound proof for emptiness also shows that the
%emptiness problem for $\cL_{>\half}(\cB)$ is in $\mathbf{\Sigma}^0_2$. Again
%this is remarkable since we show later (in Section \ref{sec:cut}) that
%the problem of deciding universality of $\cL_{>\half}(\cB)$ is not in
%arithmetical hierarchy and is in fact $\Pi^1_1$-complete.

Before we proceed with the proof of the upper bound, let us recall an
important property of finite-state B\"uchi automata
\cite{thomas-handbook,book}. The language recognized by a finite-state
B\"uchi automaton $\cA$ is non-empty iff there is a final state $q_f$
of $\cA$, and finite words $u$ and $v$ such that $\q_f$ is reachable
from the initial state on input $u$, and $\q_f$ is reachable from the
state $\q_f$ on input $v$. This implies that any
non-empty $\omega$-regular language contains an ultimately periodic
word. We had extended this observation to {\FPM}s in
\cite{rch:sis:mv:08,rch:sis:mv:08a}. In particular, we had shown that
the language $\cL_{>x}(\cM)$ is non-empty for a given $\cM$ iff there
exists a set of final states $C$ of $\cM$ and words $u$ and $v$ such
that the probability of reaching $C$ from the initial state on input
$u$ is $>x$ and for each state $q\in C$ the probability of reaching
$C$ from $q$ on input $v$ is $1$. This immediately implies that if
$\cL_{>x}(\cM)$ is non-empty then $\cL_{>x}(\cM)$ must contain an
ultimately periodic word. In contrast, this fact does not hold for
non-empty languages in $\probable$. In fact, Baier and
Gr\"{o}{\ss}er~\cite{Baier:lics2005}, construct a {\PBA} $\cB$ such
that $\cL_{>0}(\cB)$ does not contain any ultimately periodic word.

However, we will show that even though the probable semantics of a
{\PBA} may not contain an ultimately periodic, they nevertheless are
restrained in the sense that they must contain a {\it strongly asymptotic} word. 
In order to define  strongly asymptotic words formally, we introduce the following notation--

\begin{notation}
Let $\cB=(\Q,\qs,\Qf,\delta)$. Given $C\subseteq \Q$, $q\in C$ and a
finite word $u=a_0a_1\ldots a_k\in \Sigma^+$, let 
$\delta^{\Qf}_u(q,C)=\sum_{q'\in C}\delta^{\Qf}_u(q,q')$ where
$$\textstyle\delta^{\Qf}_u(q,q')=\sum_{\begin{array}{l}

 (\set q\cup \set{q'}\cup \set{q_i\st 1\leq i\leq k})\intersect \Qf \ne \emptyset
\end{array}}\delta_{a_0}(q,q_1)\delta_{a_1}(q_1,q_2)\ldots \delta_{a_k}(q_k,q').$$

\end{notation}
  
\noindent Informally, in the above notation $\delta^{\Qf}_u(q,C)$ is the
probability that the {\PBA} $\cB$, when started in state $q$, on the
input string $u$, is in some state in $C$ at the end of $u$ after
passing through a final state.  We can now define strongly asymptotic 
words.
  
 \begin{defi}
   Given a {\PBA} $\cB=(\Q,\qs,\Qf,\delta)$ and a
set $C$ of states of $\cB$, a word $\alpha\in \Sigmaw$ is said to be
{\it strongly asymptotic with respect to $\cB$ and $C$} if there is an
infinite sequence $i_1<i_2<.... $ such that 
\begin{enumerate}[(1)]
\item
$\delta_{\alpha[0:i_1]}(\qs,C)>0$ and
\item all $j>0$ and for all $q\in C$, $\delta^{\Qf}_{\alpha[i_j+1,i_{j+1}]}(q,C)>1-\frac{1}{2^j}.$
\end{enumerate}
%the probability of being in $C$ from $q$ after passing through a final
%state on the finite input string $\alpha[i_j,i_{j+1}]$ is strictly
%greater than $1-\frac{1}{2^j}.$
 A word $\alpha$ is said to be {\it
  strongly asymptotic with respect to $\cB$} if there is some $C$ such
that $\alpha$ is strongly asymptotic with respect to $\cB$ and $C.$
\end{defi}

%
% we can show that the following result holds. Given a {\PBA}
%$\cB$ and $x\in [0,1)$, $\cL_{>x}(\cB)\ne \emptyset$ iff there exists
 % a set $C$ of states of $\cB$ and a finite word $u$ such that $C$ can
 % be reached from the initial state $u$ with probability $> x$ on
  %input $u$, and there exists an {\it asymptotic sequence} with
  %respect to $C$. An asymptotic sequence with respect to $C$, is an
  %infinite input sequence $\alpha$ such that there are infinite number
  %of integers $i_1<i_2<.... $ satisfying the following property: For
  %all $j>0$, for all $q\in C$, the probability of being in $C$ from
  %$q$ after passing through a final state on the finite input string
  %$\alpha[i_j,i_{j+1}]$ converges to $1$ as $j$ tends to
  %infinity. 

%

We will now show that if the probable semantics of a {\PBA} is non-empty then 
it must contain a strongly asymptotic word. We need one more notation.
\begin{notation}
 $Reach(\cB,C,x)$ denotes the
predicate $\exists u \in\Sigma^+.\delta_{u}(\qs,C)>x.$
\end{notation}
Intuitively, the predicate $Reach(\cB,C,x)$ is true iff 
there is some finite non-empty string $u$, such that  the
probability of being in $C$  having started from the initial state $\qs$ 
and after having read $u$ is $> x$.
The existence of strongly asymptotic word in probable semantics is an immediate consequence of the following Lemma.
\begin{lem} \rm
\label{lem:increasing}
Let $\cB=(\Q,\qs,\Qf,\delta)$.  For any $x\in [0,1)$,
  $\cL_{>x}(\cB)\ne \emptyset$ iff $\exists C\subseteq \Q $ such that
  $ Reach(\cB,C,x )$ is true and  for all $j>0$ there is a finite non-empty word
  $ u_j$ such that for all $q\in C.\; \delta^{\Qf}_{u_j}(q,C) >
  (1-\frac{1}{2^j})$.
\end{lem}
\begin{proof} 
$(\Leftarrow)$ Note that it is a well-known fact that the product
  $\prod^\infty_{j=1}(1-\frac{1}{2^j})$ converges and is $>0.$ Assume
  now that $\exists C\subseteq \Q $ such that $ Reach(\cB,C,x )$ is
  true and for all $ j>0$ there is a finite word $ u_j$ such that
  $\forall q\in C\;.\delta^{\Qf}_{u_j}(q,C) >
  (1-\frac{1}{2^j})$. Since $Reach(\cB,C,x)$ is true, there is a
  finite word $u$ such that $\delta_u(\qs,C)>x.$ Fix $u$. Also for
  each $j>0$, fix $u_j$ such that $\forall q\in C$, and
  $\delta^{\Qf}_{u_j}(q,C) > (1-\frac{1}{2^j})$. % There are two cases.
%\begin{enumerate}
%\item The first case is when $x=0$. In this case, it is easy to see
%  that the infinite word $uu_1u_2\ldots$ is accepted by $\cB$ with
%  probability $>0.$
%\item The second case is when $x>0.$ In this case

 Let
  $z=\delta_u(\qs,C).$ Let $y=\frac{x}{z}.$ We have that $y<1.$ Since
   $\prod^\infty_{j=1}(1-\frac{1}{2^j})>0$ and $y<1$, there
  is a $j_0>0$ such that $\prod^\infty_{j=j_0}(1-\frac{1}{2^j})> y.$
  Now it is easy to see that the word $\alpha=
  uu_{j_0}u_{j_0+1}\ldots$ is accepted by $\cB$ with probability $>
  zy.$ But $zy$ is $x$ and the result follows.
%\end{enumerate}

$(\Rightarrow)$ Assume that $\cL_{>x}(\cB)\ne\emptyset.$ Fix an
infinite input string $\gamma \in \cL_{>x}(\cB).$ Recall that the
probability measure generated by $\gamma$ and $\cB$ is denoted by
$\mu_{\cB,\gamma}$.  For the rest of this proof we will just write
$\mu$ for $\mu_{\cB,\gamma}.$

We will call a non-empty set of states $C$ {\it good} if there is an
$\epsilon>0$, a measurable set $\Path\subseteq \Q^\omega$ of  runs, and an
infinite sequence of natural numbers $i_1<i_2 <i_3<\ldots$ such that
following conditions hold.
\begin{iteMize}{$\bullet$}
\item $\mu(\Path)\geq x+\epsilon$;
\item  For each $j>0$ and each run $\rho$ in $\Path$, we have that 
       \begin{enumerate}[(a)]
       \item $\rho[0]=\qs$, $\rho[i_j]\in C$ and
       \item at least one state in the finite sequence
         $\rho[i_j,i_{j+1}]$ is a final state.
       \end{enumerate}
\end{iteMize}
We say that a good set $C$ is \emph{minimal} if $C$ is good but for
each $q\in C$, the set $C\setminus\set{q}$ is not good. Clearly if
there is a good set of states then there is also a minimal good set of
states. 
%We have the following claim which we prove in Appendix
%\ref{app:pbaproofs}.
\\ {\bf Claim:}
\begin{iteMize}{$\bullet$}
\item There is a good set of states $C$.
\item Let $C$ be a minimal good set of states. Fix $\epsilon,\Path$
  and the sequence $i_1<i_2<\ldots$ which witness the fact that $C$ is a
  good set of states.  For each $q\in C$ and each $j>0$, let
  $\Path_{j,q}$ be the subset of $ \Path$ such that each run in
  $\Path_{j,q}$ passes through $q$ at point $i_j$, {\it i.e.}, $\Path_{j,q}=
  \set{\rho\in \Path\st \rho[i_j]=q}.$ Then there exists a $p>0$ such
  that $\mu(\Path_{j,q})\geq p$ for each $q\in C$ and each $j>0$.
\end{iteMize}
We first show how to obtain the Lemma using the above claim.  Fix a
minimal set of good states $C$. Fix $\epsilon,\Path$ and the sequence
$i_1<i_2<\ldots$ which witness the fact that $C$ is a good set of
states. We claim that $C$ is the required set of states. As
$\mu(\Path)\geq x+\epsilon$ and for each $\rho\in \Path,$
$\rho[i_1]\in C$, it follows immediately that $Reach(\cB,C,x)$. Assume
now, by way of contradiction, that there exists a $j_0>0$ such that
for each finite word $u$, there exists a $q\in C$ such that
$\delta^{\Qf}_u(q,C)\leq 1-\frac{1}{2^{j_0}}$. Fix $j_0.$ Also fix
$p>0$ such that $\mu(\Path_{j,q})\geq p$ for each $j$ and $q\in C$,
where $\Path_{j,q}$ is the subset of $ \Path$ such that each run in
$\Path_{j,q}$ passes through $q$ at point $i_j$; the existence of $p$
is guaranteed by the above claim.

We first construct a sequence of sets $L_i\subseteq \Q^+$ as follows.
Let $L_1\subseteq \Q^+$ be the set of finite words on states of $\Q$
of length $i_1+1$ such that each word in $L_1$ starts with the state
$\qs$ and ends in a state in $C.$ Formally $L_1=\set{\eta\in
  \Q^+\st |\eta|=i_1+1, \eta[0]=\qs \textrm{ and } \eta[i_1]\in C}.$
Assume that $L_r$ has been constructed. Let $L_{r+1}\subseteq \Q^+$ be
the set of finite words on states of $\Q$ of length $i_{r+1}+1$ such
that each word in $L_{r+1}$ has a prefix in $L_r$, passes through a
final state in between $i_r$ and $i_{r+1}$, and ends in a state in
$C.$ Formally, $L_{r+1}=\set{\eta\in\Q^+\st |\eta|=i_{r+1}+1,
  \eta[0:i_r]\in L_r , \exists i. (i_r<i<i_{r+1}\;\wedge\; \eta[i]\in
  \Qf)}.$
 
Note that $(L_r\Sigmaw)_{r\geq 1}$ is a decreasing sequence of measurable subsets
and $\Path\subseteq \intersect_{r>1}L_r\Sigmaw.$ Now, it is easy to
see from the choice of $j_0$ and $p$ that $\mu(L_{r+1}\Sigmaw) \leq
\mu(L_{r}\Sigmaw)- \frac{p}{2^{j_0}}.$ This, however, implies that
there is a $r_0$ such that $\mu(L_{r_0}\Sigmaw)<0.$ A contradiction.
Thus, it suffices to show that the claim is correct.

%Now, 
\noindent{\bf Proof of the claim}:
\begin{enumerate}[(1)]
\item For each $k>0,$ let $C_k=\post(\qs,\gamma[0:k]).$ Since the set
  of states $\Q$ is finite, there must be some $C$ such that $C_k=C$
  for infinitely many $k$'s. Fix one such $C$. We claim that $C$ is a
  good set of states. We need to show that $C$ satisfies the
  definition of good set of states. So we need to construct
  $\epsilon,\Path$ and the infinite sequence $i_1<i_2<\ldots$ as in
  the definition of good set of states.  We will pick $\epsilon>0$
  such that $\muacc\cB\gamma= x+ 2\epsilon.$ We construct $\Path$ and
  the sequence $i_1<i_2<\ldots $ as follows.

First let $\Path_0$ be the set of all runs starting in $\qs$ and
visiting the final states infinitely often. $\Path_0$ is measurable
and $\mu(\Path_0)= x+2\epsilon$. Take $i_1>0$ to be the smallest
integer such that $\post(\qs,\gamma[0:i_1])=C$. Inductively, assume
that we have constructed a sequence of integers $i_1 < i_2 < \cdots
< i_{j+1}$, and a measurable set $\Path_j \subseteq \Path_0$ such that
\begin{iteMize}{(a)}
\item $\mu(\Path_j) > x + \epsilon + \frac{\epsilon}{2^j}$, 
\item for each $\rho \in \Path_j$ and $k \leq j+1$, $\rho[i_k]
  \in C$, and 
\item for each $\rho\in \Path_j$ and $k < j+1$, there is
  some $i$ between $i_k$ and $i_{k+1}$ such that $\rho[i] \in
  \Qf$. 
\end{iteMize}
Observe that $\Path_0$ and $i_1$ satisfy that above conditions as
condition (c) holds vaccuously. Now for each $\ell>0$,
$\Path^\ell_j\subseteq \Path_j$ be the set of runs that visit a final
state at least one time between $i_{j+1}$ and $i_{j+1}+\ell$.
Formally, $\Path^\ell_j=\set{\rho\in \Path_j\st \exists i. ( i_{j+1}<
  i <i_{j+1}+\ell \wedge \rho[i]\in \Qf )}$. Clearly $\Path^\ell_j$ is
an increasing sequence of measurable sets and $\union_{\ell\in\Nats}
\Path^\ell_j=\Path_j$ (each run in $\Path_j$ visits the set of final
states infinitely often). Since $\mu(\Path_j) >
x+\epsilon+\frac{\epsilon}{2^j}$, there must exist a $\ell_0$ such
that
$\mu(\Path^{\ell_0}_j)>x+\epsilon+\frac{1}{2}(\frac{\epsilon}{2^j})$. Fix
$\ell_0$ and let $i_{j+2}>i_{j+1}+\ell_0$ be the smallest integer such
that $\post(\qs,\gamma[0:i_{j+2}])=C$. Let
$\Path_{j+1}=\Path^{\ell_0}_j$. It is easy to see that $\Path_{j+1}$
is measurable and that it satisfies the conditions (a), (b) and (c),
assumed inductively about $\Path_j$.

Observe that the above inductive construction ensures that
$\Path_{j+1} \subseteq \Path_j$. Take $\Path = \intersect_{j \in
  \Nats}\Path_j$. It is easy to see that $\Path$ and the sequence $i_1
< i_2 < \cdots$ constructed inductively, satisfy the claim.

\item
We have that $C$ is minimal good set of states. Note that as $C$ is
finite, we only need to show that for each $q\in Q$,
$\inf_{j>0}\mu(\Path_{j,q}) >0.$
%the function $\liminf_{j>0}\mu(\Path_j,q)=\lim_{j>0} (\inf_{k\geq j}(\mu(\Path_{k,q}))$ is $>0.$ 
We proceed by contradiction. Assume that there is some $q$ such that
$\inf_{j>0}\mu(\Path_{j,q}) =0$.  Fix one such $q$. We will obtain a
contradiction to minimality if we can show that $C\setminus \set{q}$
is also a good set of states.

In order to show that $C\setminus \set{q}$ is a good set of states, we
have to satisfy the definition of a good set of states.

Now, since $\inf_{j>0}\mu(\Path_{j,q}) =0$, there is some $j_1$ such
that $\mu(\Path_{j_1,q})<\frac{\epsilon}{4} .$ Let $\Path^1= \Path
\setminus \Path_{j_1,q}. $ We have that $\Path^1\subseteq \Path$,
$\mu(\Path^1)\geq x+ \frac{\epsilon}{2} +\frac{\epsilon}{4}$ and for
each $\rho\in \Path^1,$ $\rho[i_{j_1}]\in C\setminus\set{q}.$

Now, again as $\inf_{j>0}\mu(\Path_{j,q})=0$, there is some $j_2>j_1$
such that $\mu(\Path_{j_2,q})<\frac{\epsilon}{8} .$ Let $\Path^2=
\Path^1 \setminus \Path_{j_2,q}. $ We have that $\Path^2\subseteq
\Path^1$, $\mu(\Path^2)\geq x+\frac{\epsilon}{2} +\frac{\epsilon}{8}$
and for each $\rho\in \Path^2,$ $\rho[i_{j_2}]\in C\setminus\set{q}.$
Note also that as $j_2>j_1$ and $\Path^1\subseteq\Path$, we have that
for each each $\rho\in \Path^2$ there is some $i$ such that
$i_{j_1}<i<i_{j_2}$ and $\rho[i]\in \Qf.$

We can continue and obtain a sequence $\Path^{j_1}\supseteq
\Path^{j_2}\supseteq \ldots$ of measurable sets, and sequence
$i_{j_1}<i_{j_2}<\ldots$ such that for each $l>0$,
$\mu(\Path^{j_l})\geq x+ \frac{\epsilon}{2} +\frac{\epsilon}{2^l}$ and
for each $\rho\in \Path^{j_l},$ $\rho[i_{j_l}]\in C\setminus\set{q}.$
Furthermore for each $l>1$ and each $\rho\in \Path^l$ there is some
$i$ such that $i_{j_{l-1}}<i<i_{j_l}$ and $\rho[i]\in \Qf.$

Let $\Path'= \intersect_{l>0} \Path^l.$ We have that $\mu(\Path')\geq
x+ \frac{\epsilon}{2}$.  Clearly $\frac{\epsilon}{2}, $ $\Path'$ and
the sequence $i_{j_1}<i_{j_2}<\ldots$ witness the fact that
$C\setminus\set q$ is a good set of states. \qedhere
\end{enumerate}\medskip
\end{proof}

\noindent We get immediately that the probable semantics of a PBA, if non-empty, must contain a strongly asymptotic word. 
\begin{cor}\rm
Given a PBA $\cB$, $\cL_{>0}(\cB)\ne\emptyset$ iff $\cL_{>0}(\cB)$ contains a strongly asymptotic word. 
\end{cor}

%The proof of the above lemma implies that probable semantics of a {\PBA}
%is exactly the set of strongly asymptotic words (the
%proof can be found in Appendix \ref{app:pbaproofs}).
%\begin{cor}\rm
%\label{cor:asym}
%For any {\PBA} $\cB$, $\cL_{>0}(\cB)= \set{\alpha\in \Sigmaw\st \alpha \textrm{ is strongly asymptotic with respect to }\cB}.$

%\end{cor}

Lemma \ref{lem:increasing} also implies that emptiness-checking of
$\cL_{>0}(\cB)$ for a given a {\RPBA} $\cB$ is in $\mathbf{\Sigma}^0_2$.
\begin{cor}\rm
\label{cor:sigma02}
Given a \RPBA, $\cB$, the problem of deciding whether
$\cL_{>0}(\cB)=\emptyset$ is in $\mathbf{\Sigma}^0_2.$
\end{cor}
\begin{proof}
Let us fix a {\PBA} $\cB=(\Q,\qs,\Qf,\delta)$, and $x \in [0,1)$. Now
  Lemma~\ref{lem:increasing} says that the non-emptiness of
  $\cL_{>x}(\cB)$ is equivalent to the following property
\[
\begin{array}{rl}
\varphi = \exists C \subseteq \Q.\ \exists u \in \Sigma^*. & ((\delta_u(\qs,C) > x) \wedge\\
 & (\forall j.\ \exists u_j \in \Sigma^*.\ (\forall q \in C.\ \delta^{\Qf}_{u_j}(q,C) > (1 - \frac{1}{2^j})))
\end{array}
\]
which can be rewritten as (by moving quantifiers out)
\[
\begin{array}{rl}
\varphi = \exists C \subseteq \Q.\ \forall j.\ \exists u \in \Sigma^*.\ \exists u_j \in \Sigma^*. 
& ((\delta_u(\qs,C) > x) \wedge \\
 & (\forall q \in C.\ \delta^{\Qf}_{u_j}(q,C) > (1 - \frac{1}{2^j})))
\end{array}
\]
Now consider the property $\psi$ given as
\[
\begin{array}{rl}
\psi = \forall j.\ \exists C_j \subseteq \Q.\ \exists u \in \Sigma^*.\ \exists u_j \in \Sigma^*. 
& ((\delta_u(\qs,C_j) > x) \wedge \\
 & (\forall q \in C_j.\ \delta^{\Qf}_{u_j}(q,C_j) > (1 - \frac{1}{2^j})))
\end{array}
\]
Clearly, $\psi$ logically follows from $\varphi$. However, in our
specific case, it turns out that in fact, $\psi$ is equivalent to
$\varphi$ due to the following observations. First note, that since
there are only finitely many subsets of $\Q$, there must be a $C
\subseteq \Q$ such that $C = C_j$ for infinitely many $j$ (if $\psi$
holds). Further observe that if $\exists u_j.(\forall q \in
C.\ \delta^{\Qf}_{u_j}(q,C) > (1 - \frac{1}{2^j}))$ for some $j$ then
$\exists u_i.(\forall q \in C.\ \delta^{\Qf}_{u_i}(q,C) > (1 -
\frac{1}{2^i}))$ holds for all $i \leq j$. From these it follows that
$\varphi$ logically follows from $\psi$.

Observe that $(\delta_u(\qs,C) > x)$ and $(\forall q \in
C.\ \delta^{\Qf}_{u_j}(q,C) > (1 - \frac{1}{2^j}))$ are recursive
predicates. Thus, $\psi$ demonstrates that the non-emptiness problem
is in $\mathbf{\Pi}^0_2$, which means that emptiness is in $\mathbf{\Sigma}^0_2$.
\end{proof}

We will now show that the emptiness problem is also $\mathbf{\Sigma}^0_2$-hard.

\begin{lem}\rm
\label{lem:probhard}
Given a \RPBA, $\cB$, the problem of deciding whether $\cL_{>0}(\cB)=\emptyset$ is $\mathbf{\Sigma}^0_2$-hard.
\end{lem}

\begin{proof} 
The hardness result will be obtained by significantly modifying the
proof in~\cite{Baier:fossacs2008}, where the emptiness problem was
shown to be {\bf R.E.}-hard.

Consider a deterministic two counter machine $M$ with two counters and
a one way, read only input tape.  We can capture the computation of
$M$, as a sequence of configurations where each configuration is a
4-tuple $(q,x,a^i,b^j,m)$ where $q$ is the state of the finite state
control that $M$ changed to, $x$ is the input symbol that is read and
$i,j$ are the new counter values and $m$ indicates whether the input
head stayed in the same place, or moved right and read a new input
symbol. Here $m\in \{same,right\}$. Note that the two counter values
are represented in unary having a string of $a$s and $b$s,
respectively.  Thus, a computation of $M$ can be described by a string
over alphabet $\Sigma'$ that includes the states of $M$, the input
symbols of $M$, the symbols $a,b,same,right,(,),$ and `$,$'. In this
proof we will restrict our attention to machines $M$ that read all the
input symbols; thus, the number of steps in a computation is at least
the length of the input. A halting computation is a sequence of
configurations ending in a halting state. Define $L(M)$ to be the set
of input strings on which $M$ halts. Let $\langle M \rangle$ be a binary encoding of
$M$. Consider ${\cal H} = \{\langle M
\rangle\: |\: L(M) \neq \emptyset\}$ and ${\cal D} = \{\langle M
\rangle\: |\: L(M) \mbox{ is finite}\}$. Recall that ${\cal H}$ is
          {\bf R.E.}-complete and ${\cal D}$ is $\mathbf{\Sigma}^0_2$-complete.

Our proof of hardness will be as follows: Given a deterministic two
counter machine $M$, we will construct three {\RPBA}s $\cPP_1, \cPP_2$,
and $\cPP_3$ such that $\langle M \rangle \in {\cal D}$ iff
$\cL_{>0}(\cPP_1) \cap \cL_{>0}(\cPP_2) \cap \cL_{>0}(\cPP_3) =
\emptyset$. Since $\probable$ is closed under intersection and the
intersection of automata can be effectively
constructed~\cite{Baier:fossacs2008,GroesserThesis}, this will
demonstrate a reduction from ${\cal D}$ to the emptiness problem and
therefore prove the hardness result. Our construction of the {\RPBA}s
$\cPP_1, \cPP_2$, and $\cPP_3$ relies on ideas in~\cite{condon-lipton}
and~\cite{Baier:fossacs2008,GroesserThesis}. Hence, we begin by
recalling the key ideas from these papers that we will exploit.

For the rest of this proof, let us fix a deterministic two-counter
machine $M$ whose computations can be encoded as strings over
$\Sigma'$. Consider any rational $\epsilon$ such that
$0<\epsilon<\frac{1}{2}$. \cite{condon-lipton} give the construction
of a {\PFA} $\cR$ (that depends on $M$ and $\epsilon$) over alphabet
$\Sigma_\cR = \Sigma' \cup \{@\}$, where $@ \not\in \Sigma'$. We can
show that this {\PFA} $\cR$ satisfies the following properties.
\begin{enumerate}[(1)]
\item There exists an (computable) integer constant $d\geq 2$ such
  that if $w$ is a valid and halting computation of $M$ of length $n$,
  then the input string $(w@)^{d^{n}}$ is accepted by $\cR$ with
  probability $\geq (1-\epsilon)$; that is, the string obtained by
  concatenating $w$, $d^n$ number of times, where successive
  concatenations are separated by $@$, is accepted with probability at
  least $1 - \epsilon$.
\item Consider any input $u = w_1@w_2@\cdots @w_m@$, where no $w_i$ is
  a valid halting computation of $M$. $\cR$ accepts $u$ with
  probability at most $\epsilon$.
\end{enumerate}

The proof that $\cR$ satisfies these properties is defered to the
Appendix. Let   $\muacc\cR w$ %$Pr_\cR(w)$ 
denote the probability with which the input
$w$ is accepted by $\cR$. Observe that the above construction has the
following property: if $L(M) = \emptyset$ then any input string is
accepted by $\cR$ with probability at most $\epsilon$; on the other
hand, if $L(M) \neq \emptyset$ then there is some input that is
accepted with probability at least $1-\epsilon$.

Using the above construction of $\cR$,
\cite{Baier:fossacs2008,GroesserThesis} reduce ${\cal H}$ to the
emptiness problem, thus demonstrating its {\bf R.E.}-hardness. The
main ideas behind this are as follows. Let $\Sigma_\cPP = \Sigma_\cR
\cup \{\sharp,\$\}$, where $\sharp$ and $\$$ are symbols not in
$\Sigma_\cR$. \cite{Baier:fossacs2008,GroesserThesis} construct two
     {\RPBA}s $\cPP_1$ and $\cPP_2$ over alphabet $\Sigma_\cPP$
     such that $\cL_{>0}(\cPP_1)$ is
\[
\set{w^1_1\sharp w^1_2\sharp \cdots w^1_{k_{1}}\$\$w^2_1\sharp w^2_2\cdots w^2_{k_{2}}\$\$\cdots\st w^j_i\in \Sigma_\cR^* \textrm{ and }
 \prod_{j\geq 1}(1- (\prod^{k_j-1}_{i=1} (1-\muacc{\cR}{w^j_i})))>0}.
\]
and $\cL_{>0}(\cPP_2)$ is
\[
\set{v_1\$\$v_2\$\$\cdots\::\:v_i\in (\Sigma\cup\{\sharp\})^* \textrm{ and }
\prod_{i\geq 1}(1-(1-\epsilon)^{g(v_i)})=0}
\]
where $g(v_i)$ is the number of $\sharp$ symbols in $v_i$. Let $L_1 =
\cL_{>0}(\cPP_1)$ and $L_2 = \cL_{>0}(\cPP_2)$. The following two
observations are shown in~\cite{Baier:fossacs2008,GroesserThesis}.
\begin{enumerate}[(1)]
\item Consider any input $w = w^1_1\sharp w^1_2\sharp \cdots
  w^1_{k_{1}}\$\$w^2_1\sharp w^2_2\cdots w^2_{k_{2}}\$\$\cdots$, where
  $w^j_i\in \Sigma_\cR^*$ and $\muacc{\cR}{w^j_i} \leq \epsilon$. If $w \in
  L_2$ then $w \not\in L_1$.
\item Suppose $w_1, w_1, \ldots$ are (not necessarily distinct) words
  over $\Sigma_\cR$ such that $\muacc\cR{w_i} \geq 1-\epsilon$. For any
  $\epsilon< \frac{1}{2}$, there are $k_1, k_2, k_3, \ldots$ such that
  \[(w_1\sharp)^{k_1-1}w_1\$\$(w_2\sharp)^{k_2-1}w_2\$\$\cdots\] belongs to
  $L_1 \cap L_2$.
\end{enumerate}
Observe that the above two observations
allow~\cite{Baier:fossacs2008,GroesserThesis} to conclude that $L_1
\cap L_2 \neq \emptyset$ iff there is some $w$ such that $\muacc\cR{w}
\geq 1-\epsilon$. Thus, using properties of $\cR$, one can see that
${\cal H}$ can be reduced to the emptiness problem, therefore
demonstrating its undecidability.

In order to prove the tighter lower bound of $\mathbf{\Sigma}^0_2$, we would
like to extend the above ideas to obtain a reduction from ${\cal D}$,
instead of ${\cal H}$. We first outline the intuitions behind the
extension. Suppose $u_1, u_2, \ldots$ are (not necessarily distinct)
halting computations of $M$. Consider the input word
\[
w(k_1,k_2,\ldots) = ((u_1@)^{\ell_1}\sharp)^{k_1-1}(u_1@)^{\ell_1}\$\$
((u_2@)^{\ell_2}\sharp)^{k_2-1}(u_2@)^{\ell_2}\$\$\cdots
\]
where $\ell_i = d^{|u_i|}$. From the preceding paragraphs it can be
seen that there is a choice of $k_1,k_2, \ldots$ such that
$w(k_1,k_2,\ldots) \in L_1\cap L_2$. To obtain a reduction from ${\cal
  D}$, we need to ``check'' that infinitely many among the
computations $u_1, u_2, \ldots$ correspond to distinct inputs. To do
this we will construct a third {\RPBA} $\cPP_3$ that will check that
the computations $u_i$ grow unboundedly. Since $M$ is deterministic,
passing the test imposed by $\cPP_3$ ensures that $L(M)$ is
infinite, and conversely, if $L(M)$ is infinite then our assumption
that $M$ reads all the input symbols ensures that there will be some
string that passes the $\cPP_3$ test.

\begin{figure}
\begin{center}
\begin{tikzpicture}
\input{auto-exmp.tkz}
\end{tikzpicture}
\end{center}
\caption{Automata $\cPP_3$. Here $*$ indicates any input symbol, $b$
  any input symbol that is not $\$$, and $a$ any input symbol that is
  not $@,\sharp$ or $\$$.}
\label{fig:auto}
\end{figure}
We now outline the formal details. The {\RPBA} $\cPP_3$ has
$\Sigma_\cPP$ as input alphabet and is shown in
Figure~\ref{fig:auto}. It has four states $s_0,s_1,s_2,s_r$ where
$s_0$ is the initial state and the only final state. $s_r$ is an
absorbing state. The transition probabilities depend on a parameter
$\lambda$ that we will fix later. In state $s_0$, on inputs other than
$@,\sharp,\$$, the machine transitions to $s_0$ with probability
$\lambda$ and $s_1$ with probability $1-\lambda$; on $@,\sharp,\$$ it
goes to $s_r$ with probability $1$. In state $s_1$ it behaves as
follows. On input $@,\#$ it goes to state $s_0$ with probability $1$;
on input $\$$ it goes to $s_2$; on all other inputs, it remains in
$s_1$ with probability $1$.  In state $s_2$ it behaves as follows. On
input $\$$ it goes to state $s_0$ with probability $1$; on all other
inputs, it goes to $s_r$ with probability 1. We will pick $\lambda$ to
be such that $\lambda\cdot d \ll \frac{1}{2}$; recall that $d$ is the
constant associated with $\cR$. Let
$\SeqComp=\set{u_0x_0u_1x_1\cdots\st x_i\in \set{@,\#,\$\$} \text{ and
  } @,\#,\$ \textrm{ do not appear in the strings }u_i}$. It can be
easily shown that $\cL_{>0}(\cPP_3) \subseteq \SeqComp$. In
addition, consider $\alpha = u_0x_0u_1x_1\cdots \in \SeqComp$ with
$x_i\in \set{@,\#,\$\$}$ and $@,\#,\$$ not in $u_i$. If there is an
$\ell$ such that for infinitely many $i$, $|u_i| = \ell$ then $\alpha
\not\in \cL_{>0}(\cPP_3)$. We conclude the proof by showing the
following claim.

\noindent{\bf Claim:} $L(M)$ is a finite set iff
$\cL_{>0}(\cPP_1)\intersect \cL_{>0}(\cPP_2)\intersect\cL_{>0}(\cPP_3)
= \emptyset$.

\noindent{\bf Proof of the claim:} Let $\Cc(M)$ be the set of valid
halting computations of $M$. Since we assume that $M$ is deterministic
and reads all the input symbols, we can conclude that $L(M)$ is finite
iff $\Cc(M)$ is finite. Suppose $L(M)$ is finite and $\alpha \in
\cL_{>0}(\cPP_2)\intersect\cL_{>0}(\cPP_3)$. Since $\alpha$ is
accepted by $\cPP_3$, we know that the computations in $\alpha$ grown
unboundedly. However, since $\Cc(M)$ is finite, we can conclude that
there is a suffix $\beta$ of $\alpha$ such that none of computations
of $M$ in $\beta$ are valid and halting. Thus, $\beta = w^1_1\sharp
w^1_2\sharp \cdots w^1_{k_{1}}\$\$w^2_1\sharp w^2_2\cdots
w^2_{k_{2}}\$\$\cdots$ such that $\muacc\cR{w^j_i} \leq
\epsilon$. Coupled with the fact that $\alpha\in \cL_{>0}(\cPP_2)$, we
can argue that $\alpha\notin \cL_{>0}(\cPP_1)$ using a similar
reasoning as in~\cite{Baier:fossacs2008}.

Suppose $L(M)$ is an infinite set. Hence $\Cc(M)$ is also an infinite
set.  Let $u_1,u_2,\ldots$ be some \emph{distinct} computations in
$\Cc(M)$; we will describe how to choose $u_i$ later. As before,
consider
\[
w(k_1,k_2,\ldots) = ((u_1@)^{\ell_1}\sharp)^{k_1-1}(u_1@)^{\ell_1}\$\$
((u_2@)^{\ell_2}\sharp)^{k_2-1}(u_2@)^{\ell_2}\$\$\cdots
\]
where $\ell_i = d^{|u_i|}$. As mentioned before, there are $k_1,k_2,
\ldots$ such that $w(k_1,k_2,\ldots)$ is accepted by both $\cPP_1$ and
$\cPP_2$. In addition, the probability that $\cPP_3$ accepts
$w(k_1,k_2,\ldots)$ is $\prod_{i>0} p_i$ where
$p_i\:=\:(1-\lambda^{|u_i|})^{\ell_i\cdot k_{i}}$. We will choose
$u_i$ (or rather its length) to be a computation so that $p_i$ is
$>(1-\frac{1}{2^{i}})$; this will ensure that $\prod_{i>0} p_i$ is
non-zero. Assuming $\lambda$ to be very small 
and substituting for
$\ell_i\:=\:d^{|u_i|}$, it is easily seen that
%$p_i>(1-\ell_i\cdot \lambda^{|u_i|})^{k_i}$.  By substituting for
%$\ell_i\:=\:d^{|u_i|}$, we see that 
$p_i > (1- (d\cdot
\lambda)^{|u_i|})^{k_{i}}$. Here $d,k_i$ are fixed and $\lambda$ is a
small constant such that $d\cdot \lambda \ll \frac{1}{2}$.  Now, it
should be easy to see that we can chose a sufficiently long halting
computation $u_i$ so that $(1-(d\cdot \lambda)^{|u_i|})^{k_{i}} >
(1-\frac{1}{2^{i}})$.  
\end{proof}

Since the class $\probable$
is closed under complementation and the complementation procedure is
recursive \cite{Baier:fossacs2008} for {\RPBA}s, we can conclude that checking
universality of $\cL_{>0}(\cB)$ is also $\mathbf{\Sigma}^0_2$-complete. The same bounds also apply to
checking language containment under probable semantics. Note that these problems were
already shown to undecidable in \cite{Baier:fossacs2008}, but the exact complexity was not computed therein. 
\begin{thm}\rm
\label{thm:pbaanalytical}
Given a \RPBA, $\cB$, the problems 1) deciding whether $\cL_{>0}(\cB)=\emptyset$ and
2) deciding whether $\cL_{>0}(\cB)=\Sigmaw$, are
 $\mathbf{\Sigma}^0_2$-complete.  Given another  {\RPBA}, $\cB'$, the problem of deciding whether $\cL_{>0}(\cB)\subseteq \cL_{>0}(\cB')$  
is also $\mathbf{\Sigma}^0_2$-complete.  
\end{thm}
\begin{proof}
Since {$\probable$} is closed under complementation and  the complementation is
recursive \cite{Baier:fossacs2008,GroesserThesis}, Lemma \ref{lem:probhard} 
immediately implies that the problems of universality, emptiness
and set containment are $\mathbf{\Sigma}^0_2$-hard.
 Also observe
that given $\cB_1$ and $\cB_2$, we have that  $\cL_{>0}(\cB_1)\subseteq \cL_{>0}(\cB_2)$ iff $\cL_{>0}(\cB_1)
\intersect (\Sigmaw\setminus \cL_{>0}(\cB_2))=\emptyset.$ Now, results of \cite{Baier:fossacs2008,GroesserThesis}
show that there is a constructible $\cB_3$ such that $\cL_{>0}(\cB_3)=\cL_{>0}(\cB_1)
\intersect (\Sigmaw\setminus \cL_{>0}(\cB_2))$. Now, thanks to Corollary \ref{cor:sigma02},
the problems of universality, emptiness
and set containment are in $\mathbf{\Sigma}^0_2$.
 \end{proof}
\begin{rems}
Lemma \ref{lem:increasing} can be used to show that emptiness-checking
of $\cL_{>\half}(\cB)$ for a given {\RPBA} $\cB$ is in
$\mathbf{\Sigma}^0_2$. In contrast, we had shown in \cite{rch:sis:mv:08a} that
the problem of deciding whether $\cL_{>\half}(\cM)=\Sigmaw$ for a
given {\FPM} $\cM$ lies beyond the arithmetical hierarchy.
\end{rems}

\section{Almost-sure semantics}
\label{sec:alsure}

The class $\alsure$ was first studied in \cite{Baier:fossacs2008},
although it was not characterized topologically.  In this section,
we study the expressiveness and complexity of the class $\alsure.$ We
will also demonstrate that the class $\alsure$ is closed under finite
unions and intersections.
As in the case of probable semantics, we assume  that the alphabet $\Sigma$
is fixed and contains at least two letters.
\subsection{Expressiveness}
In this section, we shall establish new expressiveness results for the class $\alsure$--
\begin{iteMize}{$\bullet$}
\item $\alsure\subsetneq \cG_\delta.$ This is an immediate consequence of Theorem \ref{thm:closure}, Lemma \ref{lem:Gdelta} and Lemma \ref{lem:probableexp}.
\item $\Regular \intersect \alsure = \Regular \intersect \Det$ (see Proposition \ref{prop:alsureexpress}).
\item The class $\alsure$ is closed under union and intersection.  (see Corollary \ref{cor:alsurecolsure}).\medskip
\end{iteMize}

\noindent We start by characterizing the intersection $\Regular
\intersect \alsure$.  Note that the fact every language $\alsure$ is
contained in $\cG_\delta$ implies immediately that there are
$\omega$-regular languages not in $\alsure$.  That there are
$\omega$-regular languages not in $\alsure$ was also proved in
\cite{Baier:fossacs2008}, although the proof therein is by explicit
construction of an $\omega$-regular language which is then shown to be
not in $\alsure.$ Our topological characterization of the class
$\alsure$ has the advantage that we can characterize the intersection
$\Regular \intersect \alsure$ exactly: $\Regular \intersect \alsure$
is the class of $\omega$-regular languages that can be recognized by a
finite-state deterministic B\"uchi automaton.
%The following is proved in Appendix \ref{app:alsure}. 

\begin{prop}\rm
\label{prop:alsureexpress}
For any {\PBA} $\cB$, $\cL_{=1}(\cB)$ is a $\cG_\delta$ set.
Furthermore, $\Regular \intersect \alsure = \Regular \intersect \Det$
and $\Regular \intersect \Det\subsetneq \alsure \subsetneq
\cG_\delta=\Det.$
\end{prop}
\begin{proof}
Lemma \ref{lem:Gdelta}, Theorem \ref{thm:closure} and Lemma
\ref{lem:probableexp} already imply that $ \alsure \subsetneq
\cG_\delta=\Det.$ We only need to show that $\Regular \intersect
\alsure = \Regular \intersect \Det.$ Since every language in $\alsure$
is deterministic, we get immediately that $\Regular \intersect \alsure
\subseteq \Regular \intersect \Det.$ For the reverse inclusion, note
that every $\omega$-regular, deterministic language is recognizable by
a finite-state deterministic B\"uchi automaton. It is easy to see that
any language recognized by a deterministic finite-state B\"uchi
automaton is in $\alsure$. The result follows.   \end{proof}

A direct consequence of the characterization of the intersection
$\Regular \intersect \Det$ is that the class $\alsure$ is not closed
under complementation as the class of $\omega$-regular languages
recognized by deterministic B\"uchi automata is not closed under
complementation. That the class $\alsure$ is not closed under
complementation is also observed in \cite{Baier:fossacs2008}, and is
proved by constructing an explicit example.  However, even though the
class $\alsure$ is not closed under complementation, we have a
``partial'' complementation operation--- for any $\PBA$ $\cB$ there is
another $\PBA$ $\cB'$ such that $\cL_{>0}(\cB')$ is the complement of
$\cL_{=1}(\cB)$.  This also follows from the following results of
\cite{Baier:fossacs2008} as they showed that $\alsure \subseteq
\probable$ and $\probable$ is closed under complementation. However
our construction has two advantages: 1) it is much simpler than the
one obtained by the constructions in \cite{Baier:fossacs2008}, and 2)
the {\PBA} $\cB'$ belongs to the restricted class of finite
probabilistic monitors {\FPM}s (see Section \ref{sec:prelim} for
definition of {\FPM}s).  This construction plays a critical role in
our complexity analysis of decision problems.

\begin{lem}\rm
\label{lem:comp}
For any {\PBA} $\cB$, there is an {\FPM} {$\cM$} such that 
$\cL_{=1}(\cB)=\Sigmaw\setminus \cL_{>0}(\cM).$ 
\end{lem}
\begin{proof}
Let $\cB=(\Q,\qs,\Qf,\delta)$. We construct $\cM$ as follows. First we
pick a new state $\qr$, which will be the reject state of the {\FPM}
$\cM.$ The set of states of $\cM$ would be $\Q\cup\set{\qr}$.  The
initial state of $\cM$ will be $\qs$, the initial state of $\cB$. The
set of final states of $\cM$ will be $\Q$, the set of states of $\cB$.
The transition relation of $\cM$ would be defined as follows. If $q$
is not a final state of $\cB$ then the transition function would be
the same as for $\cB$. If $q$ is an final state of $\cB$ then $\cM$
will transit to the reject state with probability $\half$ and with
probability $\half$ continue as in $\cB.$ Formally, $\cM=(\Q\cup
\set{\qr}, \qs,\Q, \delta_\cM)$ where $\delta_\cM$ is defined as
follows. For each $a\in \Sigma$, $q,q'\in \Q,$
\begin{iteMize}{$\bullet$}
\item $\delta_\cM(q,a,\qr)=\half$ and $\delta_\cM(q,a,q')=\half
 \delta(\q,a,\q')$ if $q\in \Qf, $
\item $\delta_\cM(q,a,\qr)=0$ and
  $\delta_\cM(q,a,q')=\delta(\q,a,\q')$ if $q\in \Q\setminus \Qf, $
\item $\delta_\cM(\qr,a,\qr)=1$.
\end{iteMize}
It is easy to see that a word $\word\in \Sigmaw$ is rejected with
probability $1$ by $\cM$ iff it is accepted with probability $1$ by
$\cB$.  The result now follows.  
\end{proof}
The ``partial'' complementation operation has many consequences. One
 consequence is that the class
$\alsure$ is closed under union. The class $\alsure$ is easily shown
to be closed under intersection. Hence for closure properties,
$\alsure$ behave like deterministic B\"uchi automata.
%The proof of these closure properties has been deferred to Appendix~\ref{app:alsure}. 
Please note that 
closure properties were not studied in \cite{Baier:fossacs2008}.
\begin{cor}
\label{cor:alsurecolsure}
\rm
The class $\alsure$ is closed under finite union and finite intersection.
\end{cor}
\begin{proof}
Let $\cB_1 = (\Q^1,\qs^1,\Qf^1,\delta^1)$ and $\cB_2 =
(\Q^2,\qs^2,\Qf^2,\delta^2)$ be two PBAs, and we assume without loss
of generality that $\Q^1 \cap \Q^2 = \emptyset$. We will present
construction of PBAs that recognize the union and intersection of
these languages under the almost sure semantics.

We begin by first considering the construction for union. Now by
Lemma~\ref{lem:comp}, we know that there are FPMs $\cM_1$ and $\cM_2$
%$\cM_1 = (\Q^1 \cup\{\qr^1\}, \qs^1,\Q^1,\delta_{\cM_1})$ and $\cM_2 = (\Q^2 \cup
%\{\qr^2\}, \qs^2,\Q^2,\delta_{\cM_2})$ 
such that $\cL_{=1}(\cB_i) = \Sigmaw \setminus \cL_{>0}(\cM_i)$. Now,
we had shown in \cite{rch:sis:mv:08a} that there is a FPM
$\cM=(\Q,\qs,\Qf,\delta)$ such that for any word $\alpha$,
$\muacc\cM\alpha=\muacc{\cM_1}\alpha\times\muacc{\cM_2}\alpha.$ It is
easy to see that $\cL_{>0}(\cM)=\cL_{>0}(\cM_1) \cap \cL_{>0}(\cM_2).$

Now, the FPM $\cM$ can be easily ``complemented''. If $\qr$ is the
reject state of $\cM$, then consider the {\PBA} $\overline\cM =
(\Q,\qs,\set{\qr},\delta)$; clearly $\cL_{=1}(\overline\cM) =
\Sigmaw\setminus \cL_{>0}(\cM)$. Thus, by DeMorgan Laws,
$\cL_{=1}(\overline\cM) = \cL_{=1}(\cB_1) \cup \cL_{=1}(\cB_2)$.

The PBA recognizing the intersection of the languages recognized by
$\cB_1$ and $\cB_2$ with respect to almost-sure semantics does the
following: on an input $\alpha$, with probability $\frac{1}{2}$ it
runs $\cB_1$ on $\alpha$, and with probability $\frac{1}{2}$ it runs
$\cB_2$. Clearly, such a machine will accept (with respect to
almost-sure semantics) iff both $\cB_1$ and $\cB_2$ accept. Formally,
$\cB = (\Q,\qs,\Qf,\delta)$ is given by
\begin{iteMize}{$\bullet$}
\item $\Q = \Q^1 \cup \Q^2 \cup \{\qs\}$ where $\qs \not\in \Q^1 \cup
  \Q^2$
\item $\Qf = \Qf^1 \cup \Qf^2$
\item The transition relation $\delta$ is defined as follows
\begin{iteMize}{$-$}
\item For $q \in \Q^1$, $\delta(\qs,a,q) =
  \frac{1}{2}\delta^1(\qs^1,a,q)$, and for $q \in \Q^2$,
  $\delta(\qs,a,q) = \frac{1}{2}\delta^2(\qs^2,a,q)$
\item For $q,q' \in \Q^1$, $\delta(q,a,q') = \delta^1(q,a,q')$ and for
  $q,q' \in \Q^2$, $\delta(q,a,q') = \delta^2(q,a,q').$ \qedhere
\end{iteMize}
\end{iteMize}
\end{proof}

\subsection{Decision problems}
For the rest of this section, we shall focus our attention on decision
problems for almost sure semantics for {\RPBA}s. Results of this section are summarized in the
 second row of Figure \ref{fig:compresults}  on page~\pageref{fig:compresults} and proved in Theorem \ref{thm:alsurethm}
 and Theorem \ref{thm:alcontundec}.
The problem of checking whether $\cL_{=1}(\cB)=\emptyset$ for a given {\RPBA} $\cB$ was shown to
be decidable in $\exptime$ in \cite{Baier:fossacs2008}, where it was also conjectured to be 
$\exptime$-complete. The decidability of the
universality problem was left open in \cite{Baier:fossacs2008}.  We
can leverage our ``partial'' complementation operation to show that a)
the emptiness problem is in fact \pspace-complete, thus tightening the
bound in \cite{Baier:fossacs2008} and b) the universality
problem is also   \pspace-complete.

\begin{thm}\rm 
\label{thm:alsurethm}
Given a {\RPBA} $\cB$, the problem of deciding whether
$\cL_{=1}(\cB)=\emptyset$ is \pspace-complete. The problem of deciding
whether $\cL_{=1}(\cB)=\Sigmaw$ is also \pspace-complete.
\end{thm}
\begin{proof}
({\bf Upper bounds}.) We first show the upper bounds.  The proof of
  Lemma \ref{lem:comp} shows that for any {\RPBA} $\cB$, there is a
  {\RFPM} $\cM$ constructed in polynomial time such that
  $\cL_{=1}(\cB)=\Sigmaw\setminus\cL_{>0}(\cM).$ $\cL_{=1}(\cB)$ is
  empty (universal) iff $\cL_{>0}(\cM)$ is universal (empty
  respectively).  Now, we had shown in
  \cite{rch:sis:mv:08,rch:sis:mv:08a} that given a {\RFPM} $\cM$, the
  problems of checking emptiness and universality of $\cL_{>0}(\cM)$
  are in {\pspace}, thus giving us the desired upper bounds.
\hspace*{0.1cm}\\
\noindent
({\bf Lower bounds}.)  We had shown in
\cite{rch:sis:mv:08,rch:sis:mv:08a} that given a {\RFPM} $\cM$, the
problems of deciding the emptiness and universality of $\cL_{>0}(\cM)$
are $\pspace$-hard respectively. Given a {\RFPM}
$\cM=(\Q,\qs,\Q_0,\delta)$ with $\qr$ as the absorbing reject state,
consider the {\PBA} $\overline\cM=(\Q,\qs,\set\qr,\delta)$ obtained by
considering the unique reject state of $\cM$ as the only final state
of $\overline\cM$. Clearly we have that
$\cL_{>0}(\cM)=\Sigmaw\setminus \cL_{=1}(\overline\cM)$. Thus
$\cL_{>0}(\cM)$ is empty (universal) iff $\cL_{=1}(\overline\cM) $ is
universal (empty respectively). The result now follows.  
\end{proof}

Even though the problems of checking emptiness and universality of
almost-sure semantics of a {\RPBA } are decidable, the problem of
deciding language containment under almost-sure semantics turns out to
be undecidable, and is indeed as hard as the problem of deciding
language containment under probable semantics (or, equivalently,
checking emptiness under probable semantics).
\begin{thm}\rm
\label{thm:alcontundec}
Given {\RPBA}s, $\cB_1$ and $\cB_2$, the problem of deciding whether
$\cL_{=1}(\cB_1)\subseteq \cL_{=1}(\cB_2)$ is $\mathbf{\Sigma}^0_2$-complete.
\end{thm}

\begin{proof}
Observe first that given {\RPBA}s $\cB_1$ and $\cB_2$, there are
(constructible) {\RFPM}s $\cM_1$ and $\cM_2$ such that
$\cL_{=1}(\cB_i)=\Sigmaw\setminus \cL_{>0}(\cM_i)$ for $i=1,2$ (see
Lemma \ref{lem:comp}).  Thus, $\cL_{=1}(\cB_1)\subseteq
\cL_{=1}(\cB_2)$ iff $\cL_{>0}(\cM_2)\subseteq \cL_{=1}(\cM_1)$. The
upper bound then follows from the upper bound of the containment of
{\PBA}s under probable semantics.

The lower bound is shown by a reduction from emptiness-checking of
probable semantics. Recall from the proof of the fact that
$\probable=\cBool({\alsure})$ (Theorem \ref{thm:closure}) that given a
{\RPBA} $\cB$, there are {\RPBA}s $\cB^+_1, \cB^+_2,\ldots \cB^+_m$
and $\cB^-_1, \cB^-_2\ldots\cB^-_m $ such that
$$ \cL_{>0}(\cB)=\bigcup_{1\leq i\leq m} \cL_{=1}(\cB^+_i) \intersect
(\Sigmaw \setminus \cL_{=1}(\cB^-_i)).$$ Furthermore, the construction in the proof of Theorem
\ref{thm:closure} 
and results of
\cite{GroesserThesis} (Proposition \ref{prop:buchi-rabin} and the fact that 
complementation of probable semantics is a recursive operation for {\RPBA}s) implies that $\cB^+_i$ and $\cB^-_i$ are
constructible. Now, $\cL_{>0}(\cB)=\emptyset$ iff for each $i$,
$\cL_{=1}(\cB^+_i) \intersect (\Sigmaw \setminus
\cL_{=1}(\cB^-_i))=\emptyset$.  The lower bound now follows from the
observation that $\cL_{=1}(\cB^+_i) \intersect (\Sigmaw \setminus
\cL_{=1}(\cB^-_i))=\emptyset$ iff $\cL_{=1}(\cB^+_i) \subseteq
\cL_{=1}(\cB^-_i).$
\end{proof}

%%%%%%%%%%%%%%%%%%%%%%%%%%%%%%%%%%%%%%%%%%%%%%%%%%%%%%%%%%%%%%%%%%%%%%
%%%%%%%%%%%%%%%%%%%%%%%%%%%%%%%%%%%%%%%%%%%%%%%%%%%%%%%%%%%%%%%%%%%%%%

\section{Hierarchical PBAs}
\label{sec:hier}

We now identify a simple syntactic restriction on {\PBA}s which--
\begin{iteMize}{$\bullet$}
\item under
probable semantics coincide exactly with $\omega$-regular languages, and
\item  under almost-sure semantics coincide exactly with $\omega$-regular
deterministic languages.

 \end{iteMize}
 We will also establish complexity of decision problems
of emptiness and universality for the case when transition probabilities are given 
as rational numbers. The complexity results are summarized in Figure \ref{fig:hcompresults}.

  \begin{figure*}[t]
\footnotesize
\begin{center}
\begin{tabular}{| l | l | l | }

\hline
  & \multicolumn{1}{c|}{{Emptiness}} & \multicolumn{1}{c|}{{Universality}}  \\

  \hline
  
$\probableh$  &$\nl$-complete & $\pspace$-complete  \\

\hline
$\alsureh$ &$\pspace$-complete & $\nl$-complete  \\
%$$ & ploynom& 
\hline
  \end{tabular}
\end{center}  

\caption{Complexity of decision problems for {\RHPBA}s.}
\label{fig:hcompresults}
\end{figure*}

Intuitively, a hierarchical {\PBA} is a {\PBA} such that the set of
its states can be stratified into (totally) ordered levels.  From a
state $q$, for each letter $a$, the machine can transition with
non-zero probability to at most one state in the same level as $q$,
and all other probabilistic transitions go to states that belong to a
higher level. Formally,
\begin{defi}\rm
Given a natural number $k$, a {\PBA} $\cB=(\Q,\qs,\Q,\delta)$ over
an alphabet $\Sigma$ is said to be a {\it $k$-level hierarchical PBA
  ($k$-{\HPBA}}) if there is a function $\rk:\Q\to \set{0,1,\ldots, k}$
such that the following holds.
\begin{quote}
Given $j\in \set{0,1,\ldots,k}$, let $\Q_j=\set{\q\in \Q\st
  \rk(\Q)=j}.$ For every $q\in \Q$ and $a\in \Sigma$, if $j_0=\rk(q)$
then $\post(q,a)\subseteq \union_{j_0\leq \ell\leq k} \Q_\ell$ and
$|\post(q,a)\intersect\Q_{j_0}| \leq 1$.\medskip
\end{quote}

\noindent The function $\rk$ is said to be a {\it compatible ranking function}
of $\cB$ and for $q\in\Q $ the natural number $\rk(q)$ is said to be
the {\it rank} or {\it level} of $q.$ $\cB$ is said to be a {\it
  hierarchical {\PBA} ({\HPBA})} if $\cB$ is $k$-hierarchical for some
$k.$ If $\cB$ is also a {\RPBA}, we say that $\cB$ is a rational
hierarchical {\PBA} ({\RHPBA}). 
%If $\cB$ is a {\FPM} then $\cB$ is
%said to be a hierarchical probabilistic monitor ({\HPM}).  If $\cB$ is
%both $\RPBA$ and $\RHPBA$ then $\cB$ is said to be a {\RHPM}.
\end{defi}
\begin{comment}
We can define classes analogous to $\probable$, $\alsure$, $\strict$ and $\nstrict$; and we shall call them $\probableh$, $\alsureh$, $\stricth$ and $\nstricth$ respectively.
\end{comment}
We can define classes analogous to $\probable$ and $\alsure$; and we
shall call them $\probableh$ and $\alsureh$ respectively.  Before we
proceed to discuss the probable and almost-sure semantics for
{\HPBA}s, we point out two interesting facts about hierarchical
{\HPBA}s.  First is that for the class of $\omega$-regular
deterministic languages, {\HPBA}s like non-deterministic B\"uchi
automata can be exponentially more succinct.
\begin{prop}\rm
\label{prop:succinct}
Let $\Sigma=\set{a,b,c}.$ For each $n\in \Nats$, there is a
$\omega$-regular deterministic language $\sL_n\subseteq \Sigmaw$ such
that i) any deterministic B\"uchi automata for $\sL_n$ has at least
$O(2^n)$ number of states, and ii) there are {\HPBA}s $\cB_n$ s.t.
$\cB_n$ has $O(n)$ number of states and $\sL_n=\cL_{=1} (\cB_n).$
\end{prop}
\begin{proof}
Given $n\in \Nats$, let $\sL_n$ be the safety language in which for
every $a$ there is a $c$ after exactly $n$-steps.  In other words,
$\sL_n=\Sigmaw\setminus (\Sigma^*a\Sigma^n\set{a,b}\Sigmaw).$ This
could model, for instance, the property ``every request $a$ is
answered after exactly $n$-steps''. We can build a deterministic
B\"uchi automaton for $\sL_n$ and the number of states of such a
automaton is $O(2^n).$
%For $\sL_n$ we could also construct a ``non-deterministic'' B\"uchi automaton $\cA_n$ that recognizes the complement 
%$\Sigmw\setminus \SL_n$--- the non-deterministic automaton scans all inputs and upon 
%encountering $a$ non-deterministically decides to check if there is a $c$ after exactly $n$-steps. 
%This non-deterministic B\"'uchi automaton has $O(n)$ number of states. 
We could build a {\HPBA} $\cB_n$ with $O(n)$ state such that
$\sL_n=\cL_{= 1}(\cB_n)$. The {\HPBA} $\cB_n$ will be an {\FPM} also.
The construction of $\cB_n$ is as follows--- $\cB_n$ scans the input
and upon encountering $a$, $\cB_n$ decides with probability $\half$ to
check if there is a $c$ after $n$ steps and with probability $\half$,
$\cB_n$ decides to continue scanning the rest of the input. In the
former case, if the check $\cB_n$ reveals an error then $\cB_n$
rejects the input; otherwise $\cB_n$ accepts the input.  
\end{proof}

The second thing is that even though {\HPBA}s yield only
$\omega$-regular languages under both almost-sure semantics and
probable semantics, we can recognize non-$\omega$-regular languages
with cutpoints.
\begin{prop}\rm
\label{prop:nonregular}
There is a HPBA $\cB$ such that both $\cL_{\geq \half}(\cB)$ and
$\cL_{>\half}(\cB)$ are not $\omega$-regular.
\end{prop}
\begin{proof}

The {\HPBA} we will construct will actually an {\FPM}. The following
construction is given in \cite{rch:sis:mv:08}.  Let
$\Sigma=\{\zero,\one\}.$ Let $\Q=\{\q_0,\q_1,q_r\}$ and
$\delta:\Q\times\Sigma\times\Q\to \unit$ be defined as follows.  The
states $q_r$ and $q_1$ are absorbing, \ie,
$\delta(\q_r,\zero,\q_r)=\delta(\q_r,\one,\q_r)=\delta(\q_1,\zero,\q_1)=\delta(\q_1,\one,\q_1)=1$.
%For transitions out of $\q_0$,
%$\delta(\q_0,\zero,\q_0)=\delta(\q_0,\zero,\q_r)=\delta(\q_0,\one,\q_0)=\delta(%\q_0,\one,\q_1)=\frac{1}{2}.$
Transitions out of $\q_0$ satisfy
$\delta(\q_0,\zero,\q_0)=\delta(\q_0,\zero,\q_r)=\delta(\q_0,\one,\q_0)=\delta(\q_0,\one,\q_1)=\frac{1}{2}.$
Consider the {\FPM}
$\cM_\Id=\fpm{\Q,\q_0,\set{\q_0,\q_1},\delta}$. $\cM_\Id$ can be seen
to be $2$-hierarchical with $\rk(q_0)=0,\rk(q_1)=1$ and $\rk(q_r)=2.$
Given $\alpha=a_0a_1\ldots,$ it can be shown that
$\muacc{\cM_\Id}\word=\val{\alpha}$ where $\val{\alpha}$ is the real
number: $\sum_{i} \frac{\mathsf{num}(a_i)}{2^{i+1}}$ where
$\mathsf{num}(\zero)$ is the integer $0$ and $\mathsf{num}(\one)$ is
the integer $1$.

Now, consider the {\FPM} $\cM_\Id\circ \cM_\Id$ constructed as
follows.  The states of this {\FPM} are $\set{q_0,q_1}\times
\set{q_0,q_1} \cup q_{r_\new}.$ The initial state is $(q_0,q_0)$ and
the reject state is $q_{r_\new}$. The transition probabilities, from
the state $(q_{i_1},q_{j_1})$ on input $a\in \set{\zero,\one}$ is
defined as follows--- to state $(q_{i_2}, q_{j_2})$ the transition
probability is $\delta(q_{i_1},a,q_{i_2})\times
\delta(q_{j_1},a,q_{j_2})$ and to state $q_{r_\new}$ the transition
probability is $1-\sum_{i_2,j_2\in \set{0,1}}
\delta(q_{i_1},a,q_{i_2})\times \delta(q_{j_1},a,q_{j_2}).$ The state
$q_{r_\new}$ is absorbing. The {\FPM} $ \cM_\Id\circ \cM_\Id$ can be
seen to be hierarchical with $\rk(q_{i_1},q_{i_2})= i_1+i_2.$
Furthermore, it can be shown that on word $\alpha$,
$\muacc{\cM_\Id\circ \cM_\Id}\word=(\val{\alpha})^2.$ Thus,
$\cL_{>\half}({\cM_\Id\circ \cM_\Id})=\set{\word \st \val{\alpha} >
  \sqrt \half}$ and $\cL_{\geq \half}({\cM_\Id\circ
  \cM_\Id})=\set{\word \st \val{\alpha} \geq \sqrt \half}$ ; both of
which are not $\omega$-regular.  
\end{proof}

\begin{rems}
We will see shortly that the problems of deciding emptiness and
universality for a {\HPBA} turn out to be decidable under both
probable and almost-sure semantics. However, with cutpoints, they
turn out to be undecidable. 
The latter observation is out of scope of
the paper.

\end{rems}

\subsection{Probable semantics.}
We shall now show that the class $\probableh$ coincides with the class
of $\omega$-regular languages.  In~\cite{Baier:lics2005}, a restricted class of {\PBA}s
called uniform {\PBA}s was identified that also accept exactly the
class of $\omega$-regular languages. We make a couple of observations,
contrasting our results here with theirs. First the definition of
uniform {\PBA} was semantic (i.e., the condition depends on the
acceptance probability of infinitely many strings from different
states of the automaton), whereas {\HPBA} are a syntactic restriction
on {\PBA}. Second, we note that the definitions themselves are
incomparable in some sense; in other words, there are {\HPBA}s which
are not uniform, and vice versa. Finally, {\HPBA}s appear to be more
tractable than uniform {\PBA}s. We show that the emptiness problem for
$\probableh$ is \nl-complete. In contrast, the same problem was
demonstrated to be in {\exptime} and co-\np-hard~\cite{Baier:lics2005}
for uniform {\PBA}s.
%\vspace*{0.2cm}
%\noindent
%{\it Expressiveness.}

%Our main observation is that the Hierarchical {\PBA}s capture exactly
%the class of $\omega$-regular languages.
 We first establish that every
$\omega$-regular language can be recognized by a hierarchical {\PBA};
this is the content of the next Lemma.

\begin{lem}\rm
\label{lem:regisprobh}
For every $\omega$-regular language $\sL$, there is a hierarchical
{\PBA} $\cB$ such that $\sL=\cL_{>0}(\cB).$
\end{lem}
\begin{proof}

Let $\cR = (\Q,\qs,F,\Delta)$ be a deterministic Rabin automaton
recognizing $\sL$, where $F = \{(B_1,G_1),\ldots (B_k,G_k)\}$. The
hierarchical PBA will, intuitively, in the first step choose the pair
$(B_i,G_i)$ that will be satisfied in the run, and then ensure that
the measure of paths that visit $B_i$ infinitely often is
$0$. Formally, $\cB = (\Q',\qs',\Qf',\delta')$ is given as follows.
\begin{iteMize}{$\bullet$}
\item $\Q' = \{\qs',\qr'\} \cup (\{1,\ldots k\} \times \Q)$, where $\qs',\qr'
  \not\in \Q$
\item $\Qf' = \bigcup_{i=1}^k (\{i\} \times G_i)$
\item The transition relation $\delta'$ is given by
\begin{iteMize}{$-$}
\item $\delta'(\qs',a,(i,q)) = \frac{1}{k}$ iff $(\qs,a, q)\in \Delta$
\item For $q \not\in B_i$, $\delta'((i,q),a,(i,q')) = 1$ iff
  $(q,a,q')\in \Delta $
\item For $q \in B_i$, $\delta'((i,q),a,\qr') = \frac{1}{2}$ for all
  $a \in \Sigma$, and $\delta'((i,q),a,(i,q')) = \frac{1}{2}$ iff
  $(q,a,q')\in \Delta$
\item Finally, $\delta'(\qr',a,\qr') = 1$ for all $a \in \Sigma$.
\end{iteMize}
\end{iteMize}
%Observe that for each $i$ on input $\alpha$, there is only one
%run that passes through a state $(i,q)$. Furthermore runs  that visit a state in $(i,B_i)$
%infinitely often have measure $0$. Therefore,  runs that pass through a state  $(i,q)$
%that  visit $(i,G_i)$ infinitely often with strictly positive measure must visit
%$(i,B_i)$ only finitely many times.
% From this, 
 It is easy to see that $\cL_{>0}(\cB) = \sL$. Finally, we point out
 that $\cB$ is a $k+1$-level hierarchical PBA. This is witnessed by
 the ranking function $\rk$ defined as follows --- $\rk(\qs') = 0$,
 $\rk((i,q)) = i$, and $\rk(\qr') = k+1$. 
\end{proof}

\begin{thm}\rm
\label{thm:regprobh}
$\probableh=\Regular.$
\end{thm}
\begin{proof}
Thanks to Lemma \ref{lem:regisprobh} we need to show that every language
in $\probableh$ is $\omega$-regular.  The other inclusion follows from the
following Claim.
\begin{claim}\rm
%\label{lem:probhacc}
For any hierarchical {\PBA} $\cB=(\Q,\qs,\Qf,\delta)$ and any word
$\alpha\in \Sigmaw$, $\alpha\in\cL_{>0}(\cB)$ iff there is an infinite
sequence of states $\qs=q_0,q_1,\ldots$ such that $q_i\in \Qf$ for
infinitely many $i\in \Nats,$ $\delta(q_i,\alpha[i],q_{i+1})>0$ for
all $i\in \Nats$ and $\exists j\geq 0$ such that $
\delta(q_i,\alpha[i],q_{i+1})=1$ for all $i\geq j.$
\end{claim}
{\bf Proof of the claim:}
Let $\cB$ be a $k$-level hierarchical PBA with compatible ranking
function $\rk$. Let $\Q_j = \{q \in \Q\st \rk(q) =
j\}$. The proof will proceed by induction on the level $k$.

{\bf Base Case:} Suppose $k = 0$. Based on the definition of
hierarchical PBAs, this means that $\cB$ is a deterministic B\"{u}chi
automaton, i.e., for all $q,q' \in \Q$ and $a \in \Sigma$, either
$\delta(q,a,q') = 1$ or $\delta(q,a,q') = 0$. Thus, the claim clearly
holds in this case.

{\bf Induction Step:} Let $\alpha \in \Sigmaw$ be such that
$\alpha\in\cL_{>0}(\cB)$, with $\muacc\cB\alpha = x > 0$. Observe that
for every $i$, $|\post(\qs,\alpha[0,i]) \cap \Q_0| \leq 1$. There are
two cases to consider.
\begin{desCription}
\item\noindent{\hskip-12 pt\bf Case 1:}\ Suppose
  $|\post(\qs,\alpha[0,i]) \cap \Q_0| = 1$ for all $i$; let us denote
  the unique state in $\post(\qs,\alpha[0,i]) \cap \Q_0$ by
  $q_i$. Suppose in addition, there is a $j$ such that for all $\ell >
  j$, $\delta(q_\ell,\alpha[\ell],q_{\ell+1}) = 1$. Then clearly the
  sequence $q_0,q_1,\ldots$ satisfies the conditions of the lemma.
\item\noindent{\hskip-12 pt\bf Case 2:}\ Suppose Case 1 does not
  hold. Then there are two possibilities. The first possibility is
  that there is a $i_0$ such that $\post(\qs,\alpha[0,i_0]) \cap \Q_0
  = \emptyset$. The second possibility is that for every $j$, there is
  a $\ell > j$ such that $\delta(q_\ell,\alpha[\ell],q_{\ell+1}) < 1$,
  where once again we are denoting the unique state of $\Q_0$ in
  $\post(\qs,\alpha[0,\ell])$ by $q_\ell$. In this second subcase,
  there must then exist an $i_0$ such that $\delta_u(\qs,q_{i_0}) <
  x$, where $u = \alpha[0,i_0]$.

Now, based on the definition of $i_0$ given for the two subcases
above, it must be the case that for some state $q \in
\post(\qs,\alpha[0,i_0]) \setminus \Q_0$, the measure of accepting
runs from $q$ on the word $\alpha[i_0+1]\alpha[i_0+2]\cdots$ is
non-zero. Consider the hierarchical PBA $\cB' = (\Q',q,\Qf',\delta')$,
where $\Q' = \Q\setminus\Q_0$, $\Qf' = \Qf \setminus \Q_0$ and
$\delta' = \restrict{\delta}{\Q'\times\Sigma\times\Q'}$. Clearly,
$\cB'$ is a $k-1$-level hierarchical PBA, and thus by induction
hypothesis, the string $\alpha[i_0+1]\alpha[i_0+2]\cdots$ has a run
$q=q'_0q'_1\ldots$ satisfying the conditions in the claim. The desired run for
$\alpha$ (in PBA $\cB$) satisfying the conditions in the lemma is
obtained by concatenating a run from $\qs$ to $q$ on $\alpha[0,i_0]$
with $q'_0q'_1\ldots$. (\bf End proof of claim).\qed
\end{desCription}\smallskip

\noindent We now proceed with the main theorem.
Let $\cB=(\Q,\qs,\Qf,\delta).$ 
We will construct a  finite-state nondeterministic B\"uchi
automaton $\cA=(\Q',\qs',\Qf',\Delta')$, such that the language 
recognized by $\cA$ is exactly $\cL_{>0}(\cB).$
Intuitively, the set of states $\Q'$  will consist of two copies
of $\Q$--- $\Q\times \set{0}$ and $\Q\times \set{1}.$ In the first copy, we will simulate the possible transitions between pair of states of $\cB$ (we  ignore the 
exact transition probabilities of $\cB$). In the second copy, we will only simulate {\it deterministic}
transitions of $\cB$, {\it i.e.}, those transitions between pair of states which  happen with probability $1$.
From the first copy, we can transit to the second copy if the probability of transiting between the corresponding states
in $\cB$ is non-zero. From the second copy, we will never transit to the first state. The set of final states of $\cA$ are those
states in second level that correspond to the final states of $\cB.$ Intuitively, the construction ensures that if $\alpha\in \sL_{>0}(\cB)$,
 and the sequence $\qs=q_0,q_1,\ldots$ and natural number $j\geq 0$ are such that 
 \begin{enumerate}[(1)]
 \item $\q_\ell\in \Q_f$ for infinitely many $\ell$,
 \item
$0<\delta(q_i,\alpha[i],\q_{i+1})<1$ for all $i<j$ and $\delta(q_i,\alpha[i],\q_{i+1})=1$ for all $i\geq j$
\end{enumerate}
then $(\q_0,0),\dots (\q_{j},0),(\q_{j+1},1),(\q_{j+2},1)\ldots$ is an
accepting run of $\cA$ on input $\alpha.$

Formally, $\Q'$ is the set
$\Q\times\set{0,1}$, $\qs'=(\qs,0)$, $\Qf'=\set{(q,1)\st q\in \Qf}$,
and $\Delta'$ is defined as follows. For each $q_1,q_2\in \Q$,
\begin{iteMize}{$\bullet$}
\item $((q_1,0),a,(q_2,0)) \in \Delta'$  iff $\delta(q_1,a,q_2)>0.$

\item   $((q_1,0),a,(q_2,1)) \in \Delta'$  iff $\delta(q_1,a,q_2)>0.$
\item  $((q_1,1),a,(q_2,1)) \in \Delta'$  iff $\delta(q_1,a,q_2)=1.$
\item $((q_1,1),a,(q_2,0)) \in \Delta'$  iff never.
\end{iteMize}
The claim above immediately implies that $\cL_{>0}(\cB)$ is the language recognized by $\cA$ and hence is $\omega$-regular.
\end{proof}

%\begin{proof}
%Let $\cB=(\Q,\qs,\Qf,\delta).$ Consider the finite-state B\"uchi
%automaton $\cA=(\Q',\qs',\Qf',\Delta')$ where $\Q'$ is the set
%$\Q\times\set{0,1}$, $\qs'=(\qs,0)$, $\Qf'=\set{(q,1)\st q\in \Qf}$,
%and $\Delta'$ is defined as follows. For each $q_1,q_2\in \Q$,
%\begin{iteMize}
%\item $((q_1,0),a,(q_2,0)) \in \Delta'$  iff $\delta(q_1,q_2)>0.$

%\item   $((q_1,0),a,(q_2,1)) \in \Delta'$  iff $\delta(q_1,q_2)>0.$
%\item  $((q_1,1),a,(q_2,1)) \in \Delta'$  iff $\delta(q_1,q_2)=1.$
%\item $((q_1,1),a,(q_2,0)) \in \Delta'$  iff never.
%\end{iteMize}
%Lemma \ref{lem:probhacc} immediately implies that $\cL_{>0}(\cB)$ is the language recognized by $\cA$ and hence is $\omega$-regular.
%\qed
%\end{proof}

%\vspace*{0.2cm}
%{\it Decision problems.} 
We will show that the problem of deciding whether $\cL_{>0}(\cB)$ is
empty for hierarchical {\RPBA}'s is \nl-complete while the problem of
deciding whether $\cL_{>0}(\cB)$ is universal
is \pspace-complete. Thus ``algorithmically'', hierarchical {\PBA}s
are much ``simpler'' than both {\PBA}s and uniform {\PBA}s.  Note that
the emptiness and universality problem for finite state
B\"uchi-automata are also \nl-complete and \pspace-complete
respectively. 

\begin{thm}\label{thm:hpbaprobcomplex}\rm
Given a \RHPBA, $\cB$, the problem of deciding whether
$\cL_{>0}(\cB)=\emptyset$ is \nl-complete. The problem of deciding
whether $\cL_{>0}(\cB)=\Sigmaw$ is \pspace-complete.
\end{thm}
\begin{proof}
{\bf(Upper Bounds).}  First note since $\cB$ is hierarchical,
the language $\cL_{>0}(\cB)$ is $\omega$-regular (see Theorem
\ref{thm:regprobh}). The proof of Theorem \ref{thm:regprobh} also
allows us to construct a finite-state B\"uchi automata $\cA$ such that
a) $\cL_{>0}(\cB)$ is the language recognized by $\cA$ and b) the size
of the automaton $\cA$ is at-most twice the size of the automaton
$\cB.$ Furthermore, the construction can be carried out in
$\nl$. Since the emptiness problem of finite-state B\"uchi automata is
in $\nl$ and the universality problem is in $\pspace$, we immediately
get that the desired upper bounds.

{\bf(Lower Bounds).} Please note that the $\nl$-hardness of the
emptiness problem can be proved easily from the emptiness problem of
deterministic finite state machines. For the universality problem, we
make the following claim.
\begin{claim}
Given an {\FPM} $\cM$ such that the $\cM$ is also a hierarchical
{\PBA}, the problem of deciding whether $\cL_{=1}(\cM)$ is empty
is \pspace-hard.
\end{claim}
Before, we proceed to prove the claim, we first show how the lower
bound follows from the reduction.  Given an {\FPM}
$\cM=(\Q,\qs,\Qf,\delta)$ with reject state $\qr$, consider the {\PBA}
$\overline\cM=(\Q,\qs,\set\qr,\delta)$ obtained by taking the reject
state of $\cM$ as the unique final state of $\overline\cM.$ Clearly,
\begin{enumerate}[(1)]
\item $\overline\cM$ is {\HPBA} if $\cM$ is.
 \item $\cL_{>0}(\overline\cM)$ is universal iff
$\cL_{=1}(\cM)$ is empty.
\end{enumerate} 
From these two observations the desired result will follow if we can
prove the claim.  We now prove the claim.

\noindent
{\bf Proof of the claim.}
We show that there is a polynomial time bounded reduction from
every language in {\pspace} to the language 
$$\{(\cM,\Sigma)\st \cM \textrm{ is an {\FPM} on $\Sigma$, $\cM$ is a
  {HPBA} and } \cL_{= 1}(\cM)=\emptyset\}.$$ Consider a language
$L\in \pspace$ and $\cT$ be a single tape deterministic Turing machine
that accepts $L$ in space $p(n)$ for some polynomial $p$ where $n$ is
the length of its input. We assume that $\cT$ accepts an input by
halting in a specific final state $q_f$ and $\cT$ rejects an input by
not halting. Let $\cT$ be given by the tuple
$(Q,\Lambda,\Gamma,\Delta,q_0,q_f)$.  Here $Q$ is the set of states of
the finite control of $\cT$; $\Lambda,\Gamma$ are the input and tape
alphabets and $\Lambda\subset \Gamma$ and the blank symbol \# is in
$\Gamma \setminus\Lambda$; $\Delta: Q\times\Gamma\rightarrow\Gamma\times
Q\times \{Left,Right\}$; $q_0$ is the initial state and $q_f$ is the
final state. Each tuple $\Delta(q,a)=(a',q',d)$ indicates that when
$\cT$ is in state $q$, scanning a cell containing the symbol $a$, then
$\cT$ writes value $a'$ in the current cell, changes to state $q'$ and
moves in the direction $d$. Without loss of generality, we assume that
head position of $\cT$ initially is at cell number $0.$

Let $\Phi'\;=\Gamma \times Q$ and
$\Phi\;=\Phi'\cup \Gamma$. We call members of $\Phi'$ composite symbols.
A configuration of $\cT$, on an input
of length $n$, is a string of symbols, of length $p(n)$, drawn from $\Phi$.
We can define a valid configuration in the standard way.
% \cite{ahu}. 
In each valid configuration there can be only one composite symbol 
(i.e.,from $\Phi'$)
and that indicates the head position of $\cT$.
A computation
of $\cT$ is a sequence of configurations which is either finite or infinite
depending on whether the input is accepted or not. A computation
 starts in an initial configuration and each succeeding configuration
is obtained by one move of $\cT$ from the previous configuration. The initial
configuration contains the input string and the first symbol in it is from
$\Gamma\times Q$ indicating its  head position is on the first cell.

For given input $\sigma$, we construct a $\FPM$ $\cM_\sigma$ such that $\cM_\sigma$ is a $2$-{\HPBA} and
$\cM_{\sigma}$ accepts some infinite input with probability $1$
iff $\cT$ rejects $\sigma$, i.e., $\cT$ does not halt on $\sigma$.
Let $\sigma$ be an input to $\cT$ of length $n$ and let $m=p(n)$.
A state of the automaton $\cM_{\sigma}$ is a pair of the form
$(i,s)$ where $0\leq i< m$ and $s\in \Phi$, or is in $\{\qs,\qr\}$;
here  $\qs$  is the
initial state and is of rank $0$ and $\qr$ is the reject state and is of rank $2$.
The rank of states $\set{(i,s)\st 0\leq i< m \textrm{ and }s\in \Phi}$ will be $1$.
 Intuitively, if $\cM_{\sigma}$
is in state $(i,s)$ that denotes that $i^{th}$ element of the current 
configuration of the computation of $\cT$ has value $s$. Note that $s$ is
in $\Phi'$ or is in $\Gamma$. 
The input alphabet to $\cM_{\sigma}$ is the set 
$\{0,...,m-1\}\times\Phi'\times \{left,right\}$ 
together with an additional input symbol $\tau$;
that is each input to the automaton is $\tau$ or is of the form $(i,(b,q),d)$.

Let $\sigma\;=\;\sigma_0,...,\sigma_{n-1}$ be the input to $\cT$.
 The transitions of $\cM_{\sigma}$ are defined as follows.
From the initial state $\qs$, on input $\tau$, there are 
transitions to the states $(i,r_i)$, for each $i\in \{0,...,m-1\}$ where,
$r_i=(\sigma_0,q_0)$ and
$r_i=\sigma_i$ for $0<i<n$, and is the blank symbol otherwise; the probability
of each of these transitions is $\frac{1}{m}$.  
 Thus the input $\tau$ sets up the initial configuration when 
$\cM_{\sigma}$ is in the initial state $\qr$. From every other state on input
$\tau$ there is a transition to the reject state $\qr$ with probability $1$.
Also, from the initial state $\qs$, there  is a transition to the reject state $\qr$ with probability $1$
for all input symbols other than $\tau.$

From any state of the form $(j,(b,q))$ on input of the form $(i,(a,q'),d)$ the transition is defined
as follows: if  
$i=j$, $q=q'$, $b=a$ and $\Delta(q,a)=(q_1,a_1,d)$, then there is a transition to the automaton state $(j,a_1)$;
otherwise, the transition is to $\qr$; in either case, the probability
of the transition is $1$. Note that if $\cT$ halts then also there is a 
transition to $\qr$.

From any state of the form $(j,b)$, where $b\in \Gamma$, on input symbol of the form
$(i,(a,q),d)$ the transitions are defined as follows: 
if  either
$i=j-1$, $d=right$ and $\Delta(q,a)=(q',a',d')$, or  if $i=j+1$, $d=left$ and $\Delta(q,a)=(q',a',d')$ then
the transition is to the state $(j,(b,q'))$; otherwise the transition is
back to $(j,b)$; in both cases the probability of the transition is $1$.

Suppose $\sigma$ is rejected, i.e., $\cT$ does not terminate on $\sigma$.
Furthermore assume that the composite symbols in each successive configuration of the infinite computation
of $T$ on input $\sigma$
are $(a_0,q_0),(a_1,q_1),...$ and they occur in positions $i_0,...$ and the 
direction of the head movement is given by $d_0,...$
respectively. Then $\cM_{\sigma}$ accepts the infinite string $\tau(i_0,(a_0,q_0),d_0),..,(i_k,(a_k,q_k),d_k),...$
with probability $1$ and accepts all others with probability less than $1$.
It is not difficult to see that if $\sigma$ is accepted by $\cT$, all
input strings are accepted by $\cM_{\sigma}$ with probability less than
$1$. The above reduction is clearly polynomial time bounded.   
{\bf (End proof of the claim.)}
\end{proof}

%\begin{proof}
%We had shown in \cite{techreport2} that given a rational hierarchical
%probabilistic monitor $\cM$, the problems of deciding the emptiness
%and universality of $\cL_{>0}(\cM)$ are $\nl$-hard and $\pspace$-hard
%respectively. Since any $\RHPM$ is a $\RHPBA$, we get the desired
%lower bounds.

%We now show the upper bounds. First note since $\cB$ is hierarchical,
%the language $\cL_{>0}(\cB)$ is $\omega$-regular (see Theorem
%\ref{thm:regprobh}). The proof of Theorem \ref{thm:regprobh} also
%allows us to construct a finite-state B\"uchi automata $\cA$ such that
%a) $\cL_{>0}(\cB)$ is the language recognized by $\cA$ and b) the size
%of the automaton $\cA$ is at-most twice the size of the automaton
%$\cB.$ Furthermore, the construction can be carried out in
%$\nl$. Since the emptiness problem of finite-state B\"uchi automata is
%in $\nl$ and the universality problem is in $\pspace$, we immediately
%get that the desired upper bounds.
%\qed
%\end{proof}

\subsection{Almost-sure semantics.}
For a hierarchical PBA, the ``partial'' complementation operation for almost-sure semantics discussed in
 Section \ref{sec:alsure} 
yields a hierarchical PBA. Therefore using Theorem \ref{thm:regprobh}, we
immediately get that a language $\cL\in \alsureh$ is
$\omega$-regular. Thanks to the topological characterization of $\alsureh$ as a sub-collection of 
deterministic languages, we get that $\alsureh$ is exactly the class of languages recognized by deterministic finite-state B\"uchi automata.
\begin{thm}\rm
\label{thm:alsurehdet}
$\alsureh=\Regular \intersect \Det .$
\end{thm}
\begin{proof}
The inclusion $\Regular \intersect \Det \subseteq \alsureh$ follows
immediately from the fact that any language in $\Regular\intersect
\Det$ is recognizable by a finite-state deterministic B\"uchi
automaton. For the reverse inclusion $\alsureh\subseteq \Regular \intersect \Det $, note that since
$\alsure\subseteq \Det,$ it suffices  to show that $\alsureh\subseteq \Regular.$ 
Now, given $\sL\in \alsureh,$ Lemma \ref{lem:comp}
immediately implies that there is an {\FPM} $\cM$ such that
$\cL_{>0}(\cM)=\Sigmaw\setminus \sL.$ Furthermore, it is easy to see
from the proof of Lemma \ref{lem:comp} that we can take $\cM$ to be
hierarchical given that $\sL\in \alsureh.$ Now, thanks to Theorem
\ref{thm:regprobh}, $\cL_{>0}(\cM)$ is $\omega$-regular which implies
that $\sL$ is also $\omega$-regular. 
\end{proof}

\begin{comment}
\begin{proof}
The inclusion $\Regular \intersect \Det \subseteq \alsureh$ follows
immediately from the fact that any language in $\Regular\intersect
\Det$ is recognizable by a finite-state deterministic B\"uchi
automaton. Observe that, from Proposition~\ref{prop:reg-det} and the
fact that $\alsureh\subseteq\cG_\delta=\Det$, it follows that if we
show that $\alsureh \subseteq \Regular$, then $\alsureh subseteq \Regular
\intersect \Det$. Now, given $\sL\in \alsureh,$ Lemma \ref{lem:comp}
immediately implies that there is an {\FPM} $\cM$ such that
$\cL_{>0}(\cM)=\Sigmaw\setminus \sL.$ Furthermore, it is easy to see
from the proof of Lemma \ref{lem:comp} that we can take $\cM$ to be
hierarchical given that $\sL\in \alsureh.$ Now, thanks to Theorem
\ref{thm:regprobh}, $\cL_{>0}(\cM)$ is $\omega$-regular which implies
that $\sL$ is also $\omega$-regular.  
\end{proof}
\end{comment}
The ``partial'' complementation operation also yields the complexity
of emptiness and universality problems.
\begin{thm}\rm
\label{thm:alsurecomplex}
Given a \RHPBA, $\cB$, the problem of deciding whether
$\cL_{=1}(\cB)=\emptyset$ is \pspace-complete. The problem of deciding
whether $\cL_{=1}(\cB)=\Sigmaw$ is \nl-complete.
\end{thm}
{\bf (Upper Bounds.)} The upper bounds
are obtained by constructing the {\FPM} $\cM$ such that
$\cL_{=1}(\cB)=\Sigmaw\setminus \cL_{>0}(\cM)$ as in the proof of
Lemma \ref{lem:comp}. Now, $\cM$ is hierarchical if $\cB$ is hierarchical. The result now follows
immediately from  Theorem
\ref{thm:hpbaprobcomplex}.  
\begin{proof}
{\bf(Lower Bounds.)}
The $\nl$-hardness of checking universality can be shown from $\nl$-hardness of checking emptiness of deterministic
finite state machines.  
Please recall that in the proof of Theorem \ref{thm:hpbaprobcomplex}, we had shown that given an {\FPM} $\cM$
such that $\cM$ is a {\HPBA}, the problem of checking whether $\cL_{=1}(\cM)$ is empty is {\pspace}-hard.
Thus, it follows immediately that checking emptiness of $\cL_{=1}(\cB)$ for a {\HPBA} is {\pspace}-hard.
\end{proof}

\section{Conclusions}
\label{sec:conc}
In this paper, we  investigated the power of randomization in finite state
automata on infinite strings. We presented a number of results on the 
expressiveness and decidability problems under different notions of acceptance
based on the probability of acceptance. In the case of decidability, we
gave tight bounds for both the universality and emptiness problems.
%As part of future work, it will be interesting to investigate the 
%power  of Rabin acceptance for almost sure semantics. 
As part of future work, it will be
interesting to investigate the power of randomization in other models of 
computations on infinite strings such as pushdown automata etc. 
Since the universality and emptiness problems are PSPACE-complete for
almost-sure semantics, their application to practical systems needs further
enquiry.

\section*{Acknowledgements.}
The authors thank anonymous referees whose comments have improved the presentation of the paper. 
%Rohit Chadha was supported in part by NSF grants CCF04-29639 and NSF CCF04-48178. 
%A. Prasad Sistla was supported in part by NSF CCF-0742686. Mahesh Viswanathan was 
%supported in part by NSF CCF04-48178 and NSF CCF05-09321. 

%\begin{thebibliography}{condon-lipton}
%\bibliography{ref}
 %\end{thebibliography}
\bibliographystyle{alpha}
\bibliography{ref}

\appendix
\section{Properties of $\cR$ in the proof of Lemma \ref{lem:probhard}}
\begin{lem}\rm
\label{lem:condon-lipton-prop}
Let $M$ be a deterministic 2-counter machine with a one way read only input tape whose configurations are
encoded over alphabet $\Sigma'$. Let $\epsilon$ be any rational such
that $0<\epsilon<\frac{1}{2}$. There is a {\PFA} $\cR$ over alphabet
$\Sigma_\cR = \Sigma' \cup \{@\}$, where $@ \not\in \Sigma'$, such
that
\begin{enumerate}[(1)]
\item There exists an (computable) integer constant $d\geq 2$ such
  that if $w$ is a valid and halting computation of $M$ of length $n$,
  then the input string $(w@)^{d^{n}}$ is accepted by $\cR$ with
  probability $\geq (1-\epsilon)$, and
\item Any input $x = w_1@w_2@\cdots @w_m@$, where no $w_i$ is a valid
  halting computation of $M$, is accepted by $\cR$ with probability at
  most $\epsilon$.
\end{enumerate}
\end{lem}
\begin{proof}
We begin by recalling some details of the construction given
in~\cite{condon-lipton}. Given $M$ and $\epsilon_0$,
\cite{condon-lipton}\footnote{The construction in \cite{condon-lipton} is actually carried out only for deterministic 2-counter machines without an input tape. However, the construction easily carries over to deterministic 2-counter machines with one-way read only input tape.} give the construction of a {\PFA}
$\cR^M(\epsilon_0)$ which expects the input to be of the form $x =
w_1@w_2@\cdots w_m@$, where $w_i \in (\Sigma')^*$. The automaton
$\cR^M(\epsilon_0)$ tries to check that each $w_i$ is a valid, halting
computation of $M$. Since $\cR^M(\epsilon_0)$ has only finitely many
states, it cannot reliably check consistency as it requires
maintaining counter values. Instead $\cR^M(\epsilon_0)$ plays a game
on reading a computation $w_i$, wherein it tosses $O(n)$ coins; here
$n$ is $|w_i|$. The game has four possible outcomes.
\begin{enumerate}[(1)]
\item \emph{reject}, when $\cR^M(\epsilon_0)$ discovers an error in
  $w_i$,
\item \emph{double wins}, 
\item \emph{sum wins}, or
\item \emph{neither wins}
\end{enumerate}
Details of how this game is played are beyond the scope of this paper,
and the interested reader is referred to~\cite{condon-lipton}. If
$w_i$ is a valid halting computation, the following properties are
known to hold: (a) \emph{reject} is never an outcome of the game, (b)
the probability of outcome \emph{double wins} is equal to the
probability of the outcome \emph {single wins}, which we will denote
by $p$ in this proof, and (c) $p \geq 2^{-4n}$. The automaton
$\cR^M(\epsilon_0)$ maintains two counters $D$ and $S$ that take
values between $0$ and $q$ --- $q$ is a constant integer whose value
will be fixed later in the next paragraph. After playing the game on
$w_i$, $\cR^M(\epsilon_0)$ takes the following steps depending on the
outcome of the game. If the outcome is \emph{reject}, then
$\cR^M(\epsilon_0)$ moves to a special reject state $q_r$ and ignores
the rest of the input. If the outcome is \emph{double wins} then
counter $D$ is incremented, and if the outcome is \emph{sum wins} then
counter $S$ is incremented. When the outcome is \emph{neither wins},
the counters $D$ and $S$ are left unchanged. The automaton
$\cR^M(\epsilon_0)$ then checks the value of $D$ and $S$ --- if either
of them are $q$ then it ignores the rest of the input and does not
play anymore games; on the other hand if both $S$ and $D$ are less
than $q$ then it processes $w_{i+1}$ by playing the game.

After processing the entire input $x = w_1@w_2@\cdots w_m@$, the
automaton $\cR^M(\epsilon_0)$ decides to accept or reject $x$ as
follows.
\begin{enumerate}[(1)]
\item If $\cR^M(\epsilon_0)$ is in the reject state $q_r$ (i.e., one
  of the games played had outcome \emph{reject}) then $x$ is rejected.
\item $x$ is also rejected if either (a) both $S$ and $D$ are $< q$,
  or (b) $D = q$ and $S = 0$.
\item In all other cases, $x$ is accepted, i.e., when $S = q$ or when
  $D = q$ and $S \neq 0$.
\end{enumerate}
The constant $q$ is fixed to ensure that the following property holds:
Assuming that after processing $x$ at least one of the counters $D$ or
$S$ is $q$, the probability that $x$ is accepted is (a) $> 1 -
\epsilon_0$ if all the $w_i$s are valid, halting computations of $M$,
and (b) $\leq \epsilon_0$ if none of the $w_i$s are valid, halting
computations. 

In proving this lemma, we will take $\cR$ to be the {\PFA} obtained by
taking $\epsilon_0 = \frac{\epsilon}{2}$, i.e., $\cR =
\cR^M(\frac{\epsilon}{2})$. We will first show that any input $x =
w_1@w_2@\cdots @w_m@$, where no $w_i$ is a valid halting computation
of $M$, is accepted by $\cR$ with probability at most $\epsilon$. Now,
we know by construction of $\cR$ and properties stated above, assuming
that one of $S$ or $D$ is $q$, the probability that $x$ is accepted is
at most $\frac{\epsilon}{2}$. Moreover, when both $S$ and $D$ are $<
q$, we know (by construction) that $\cR$ rejects $x$. Thus the second
condition in the lemma holds.

We will now prove the first condition. Consider input $x =
w_1@w_2@\cdots @w_m@$, where all the $w_i$s are valid, halting
computations of $M$ of (equal) length $n$. Taking $p_{SD}$ to be the
probability that either $D$ or $S$ is $q$ after processing $x$, the
probability that $x$ is accepted by $\cR$ is at least $p_{SD}(1 -
\frac{\epsilon}{2})$. Thus, to prove the first condition, all we need
to show is that if $m$ is larger than $d^{n}$, for some fixed,
computable $d$, then $p_{SD}(1 - \frac{\epsilon}{2}) > 1 -
\epsilon$. In other words, we need to prove that for such large $m$,
$p_{SD} > \delta$, where $\delta =
(1-\epsilon)/(1-\frac{\epsilon}{2})$. 

Let $\overline{p_{SD}} = 1 - p_{SD}$. So $\overline{p_{SD}}$ is the
probability that both $S$ and $D$ are less than $q$ after $x$ is
processed. Recall that $p$ is the probability that $D$ is incremented
after a single computation $w_i$ is processed by $\cR$. Moreover,
since all the $w_i$s are assumed to be valid computations, $p$ is also
the probability that $S$ is incremented after playing one game. Thus,
we can say that
\[
\overline{p_{SD}} = \sum_{0 \leq S < q}\sum_{0 \leq D <
  q}\comb{m}{S+D}\comb{S+D}{S}p^{S+D}(1-2p)^{m - (S+D)}
\]
where $\comb{k}{\ell}$ is the number of ways of choosing $\ell$
objects from $k$ objects. We can simplify the above expression as
follows.
\[
\begin{array}{rl}
\sum_{0 \leq S < q}\sum_{0 \leq D < q}& \comb{m}{S+D}\comb{S+D}{S}p^{S+D}(1-2p)^{m
  - (S+D)} \\ 
 & \leq \comb{2q-2}{q-1}\sum_{0 \leq k \leq  2q-2}
               \min(k,q-1)\comb{m}{k}p^k(1-2p)^{m-k}\\ 
 & = \comb{2q-2}{q-1}\sum_{0 \leq k \leq 2q-2}
               \frac{\min(k,q-1)}{2^k}\comb{m}{k}(2p)^k(1-2p)^{m-k}\\ 
 & \leq \comb{2q-2}{q-1}\sum_{0 \leq k \leq 2q-2}\comb{m}{k}(2p)^k(1-2p)^{m-k}
\end{array}
\]
In the above reasoning, the second line follows from the observation
that since $S+D \leq 2q-2$, $\comb{S+D}{S} \leq \comb{2q-2}{q-1}$, and
the last line follows from the fact that $\frac{\min(k,q-1)}{2^k} \leq
1$. Also observe that $\sum_{0 \leq k \leq
  2q-2}\comb{m}{k}(2p)^k(1-2p)^{m-k}$ is nothing but the cumulative
distribution function for a binomial distribution with parameters $2p$
and $m$. Taking $m$ such that $2q -2 < m(2p)$, we upper bound the
above expression using Chernoff bounds as follows,
\[
\overline{p_{SD}} \leq \comb{2q-2}{q-1}\sum_{0 \leq k <
  2q-2}\comb{m}{k}(2p)^k(1-2p)^{m-k} \leq \comb{2q-2}{q-1}{\rm
  exp}(-\frac{(2pm - (2q-2))^2}{4pm})
\]
Let $\rho = \comb{2q-2}{q-1}$. Now $\rho\cdot\:{\rm exp}(-\frac{(2pm -
  (2q-2))^2}{4pm}) < 1- \delta$ when $m > \frac{\theta + (2q-2)}{p}$,
where $\theta = \log (\frac{\rho}{1 - \delta})$. Finally since $p \geq
2^{-4n}$, we get the desired $d$  for the lemma.  
\end{proof}

\end{document}